\documentclass[journal]{IEEEtran}

\usepackage{amsmath,epsfig,amssymb}
\usepackage[ruled,vlined]{algorithm2e}
\usepackage{psfrag}
\usepackage{subfigure}
  \usepackage{float}
  \usepackage{cases}
\usepackage[amsmath,thmmarks]{ntheorem}
\usepackage{color}
\usepackage{tikz}
\usepackage{graphicx}

 \floatstyle{ruled}
 \newfloat{algorithm}{htb}{loa}
 \floatname{algorithm}{\textbf{Algorithm}}
 \makeatletter
 \providecommand*{\toclevel@algorithm}{0} 
 \makeatother

\theoremsymbol{$\Box$}
\newtheorem{remark}{Remark}

\theoremsymbol{$\Box$}
\newtheorem{definition}{Definition}

\theoremsymbol{$\Box$}
\newtheorem{lemma}{Lemma}

\theoremsymbol{$\Box$}

\theoremsymbol{$\Box$}
\newtheorem{proposition}{Proposition}

\theoremsymbol{$\Box$}

\theoremsymbol{$\Box$}

\makeatletter
\newcommand\unmarkfootnote[1]{%
  \begingroup
    \let\@makefntext\noindent
    \footnotetext{#1}%
  \endgroup}
\makeatother

\makeatletter
\newcommand{\rmnum}[1]{\romannumeral #1}
\newcommand{\Rmnum}[1]{\expandafter\@slowromancap\romannumeral #1@}
\makeatother

\newcommand{\mat}{\boldsymbol}

\usepackage{xspace}
\usepackage{setspace}

\begin{document}

%
\title{Pareto Boundary of the Rate Region for Single-Stream MIMO Interference Channels: Linear Transceiver Design}
\author{Pan Cao, \IEEEmembership{Student Member, IEEE},
        Eduard A. Jorswieck, \IEEEmembership{Senior Member, IEEE}
        and Shuying Shi 
\thanks{The work of P. Cao is supported by the China Scholarship Council (CSC). The work of E. A. Jorswieck has been performed in the framework of the European research project DIWINE, which is partly funded by the European Union under its FP7 ICT Objective 1.1 - The Network of the Future. Part of this work was presented at the 13th IEEE International Workshop on Signal Processing Advances in Wireless Communications, Cesme, Turkey, June 2012.

P. Cao and E. A. Jorswieck are (S. Shi was) with the Chair for Communications Theory, Communications Laboratory, Dresden University of Technology, Dresden 01062, Germany (e-mail: \{Pan.Cao, Eduard.Jorswieck, Shuying.Shi\}@tu-dresden.de) }
}

\markboth{IEEE TRANSACTIONS ON SIGNAL PROCESSING, VOL.~XX, NO.~XX, MONTH~YEAR}{}%

\maketitle

\begin{abstract}
We consider a multiple-input multiple-output (MIMO) interference channel (IC), where a single data stream per user is transmitted and each receiver treats interference as noise.
The paper focuses on the open problem of computing the outermost boundary (so-called Pareto boundary-PB) of the achievable rate region under linear transceiver design. The Pareto boundary consists of the strict PB and non-strict PB. For the two user case, we compute the non-strict PB and the two ending points of the strict PB exactly. For the strict PB, we formulate the problem to maximize one rate while the other rate is fixed such that a strict PB point is reached. To solve this non-convex optimization problem which results from the hard-coupled two transmit beamformers, we propose an alternating optimization algorithm. Furthermore, we extend the algorithm to the multi-user scenario and show convergence. Numerical simulations illustrate that the proposed algorithm computes a sequence of well-distributed operating points that serve as a reasonable and complete inner bound of the strict PB compared with existing methods.
\end{abstract}


\section{Introduction}\label{sec:intro}

\IEEEPARstart{I}{n} wireless cellular systems, multiple sectors of different cells share the same time-frequency resource for communication in order to increase the spectral efficiency and occupancy level, while inter-cell interference brings a strong inter-cell coupling and limits the performance. In this paper, we consider first a two-cell environment, where each cell has a base station (BS) with multiple antennas and a mobile station (MS) with multiple antennas. Each BS is intended to communicate with the MS in its own cell while simultaneously interferes the MS in the other cell. And then a multi-cell scenario consisting of multiple interfering BS-MS links will be also considered later.
These scenarios are modeled as a two or multi-user MIMO interference channel (IC). The IC is characterized by its capacity region, defined as the set of largest rates that can be simultaneously achieved by the users in the system while making the error probability arbitrary small. A pragmatic approach that leads to an achievable region or inner bound of the capacity region is based on two assumptions. \rmnum{1}) The class of encoding strategies are constrained to use random Gaussian codebooks; \rmnum{2}) The decoders are restricted to treat the interference as Gaussian noise. Herein, based on these two assumptions,  we desire to find the \emph{complete achievable rate region} by linear transceiver design.

Form the perspective of optimization, it is well-known that a bi/multi-objective optimization problem usually admits infinite number of \emph{noninferior} solutions (theoretical limits), which form the outermost boundary of achievable performance region, so-called \emph{Pareto boundary} \cite{ParetoB}. A noninferior solution on the Pareto boundary is considered to be Pareto-optimal in the sense that no other solution can improve the performance of some objectives without reducing other objective(s). Generally, it is hard to find the Pareto boundary efficiently,
but it is significant to study it in order to determine optimal system operations based on Pareto-optimal rate tuples and their associated strategies.
In this paper, we propose an algorithm to compute the complete Pareto boundary{\footnote{When referring to Pareto boundary, we mean the Pareto boundary of achievable rate region unless otherwise specified hereafter.}} for the two/multi-user MIMO single-stream IC through linear transceiver design.

How to design linear transceiver schemes to achieve the Pareto boundary has attracted intensive research for several decades. A brief, comprehensive, yet non-exhaustive review of the related works is given as follows.

\textbf{Parameterization Approaches:} The Pareto boundary is characterized by a few parameters. For a two-user  multiple-input single-output (MISO) IC, the authors proposed a necessary condition for Pareto-optimal transmit beamformers, i.e., linear combinations of zero-forcing (ZF) and maximum-ratio transmission (MRT) beamformers with two $[0, 1]$-parameters \cite{EduardMISOCharcter}. This parameterization is later used to derive a characterization with only a single parameter in \cite{Closed-FormMISO}, \cite{WalrasianEqRami}. A general framework for parameterizing Pareto-optimal transmit strategies for multi-user MISO IC was proposed in \cite{RamiTSP2011}, which is applicable when the utility functions of the systems are monotonic in the received power gains. In \cite{CharacterizeMIMOSimplereceiver}, the authors proposed a parametrization to characterize the Pareto boundary of the multi-cell MIMO performance region under an assumption that each receiver has only a single effective antenna, while this limiting assumption, in fact, degrades the MIMO IC to the MISO IC.

\textbf{Computation Approaches:}
Different from the parameterization schemes, another approach is to compute a Pareto boundary point directly.
One branch is to maximize the rate of one user for a fixed rate of the other users (e.g., in two-user MISO IC \cite{EfficientComputationMISO} and in two-user MIMO IC \cite{CaoSPAWC2012}). Another important branch is to find the intersection point between the Pareto boundary and a ray from the origin (e.g., so-called rate profiles approach in the multi-user MISO IC \cite{RZhangMISO}, \cite{RZhangMISO2}). The computations approaches are able to compute the \emph{whole} Pareto boundary, if the optimization problems are solved optimally.

\textbf{Weighted Sum Maximization Approaches:}
A standard technique for generating the Pareto-optimal solutions to multi-objective optimization problems is to maximize weighted sums of the different objectives for various different settings of the weights. Generally, the weighted sum rate maximization problem for multi-user IC is non-convex.
For the multi-user SISO IC, the  MAPEL algorithm proposed in \cite{MAPEL} transformed the weighted sum rate maximization into a generalized linear fractional programming problem, which can be solved optimally. In \cite{WSRMISO2012}, the authors jointly utilized the monotonic optimization and rate profile techniques to solve the weighted sum rate maximization optimally at the cost of computation load in multi-user SISO/MISO/SIMO IC.
However, it is NP-hard to obtain a global optimal solution of the weighted sum rate maximization for a multi-user MIMO IC (e.g., in \cite{MIMOIBCLuo}). In particular, most algorithms focus on finding only a single sum-rate maximum point, e.g., for the two-user MIMO IC based on pricing in \cite{InterferencePricingMIMO},  for the two-user MIMO IC based on approximation of sum rate in \cite{Two-cellMIMO}, for the single-stream MIMO IC based on balancing the egoistic and the altruistic behavior in \cite{BalancingMIMO}, and for the multi-user MIMO IC based on interference alignment in \cite{CoorpertiveMIMOIA}. However, it is well-known that the weighted sum maximization method has two major drawbacks \cite{WSODrawbacks}: \rmnum{1}) If the Pareto boundary is not convex, there does not exist any weight corresponding to the points on the nonconvex part. Increasing the number of steps of the weighting factor does not resolve this problem; \rmnum{2}) Even if the Pareto boundary is convex, an even spread of weights does not produce an even spread of points on the Pareto boundary. Therefore, weighted sum rate maximization is not a promising method to perform \emph{the complete Pareto boundary}, especially the non-convex boundary.

\textbf{Game Theoretic Approaches:}
Game theory as a useful tool has been widely applied to resource allocation in multi-user IC by studying the conflicting or cooperative behavior of the users. Distributed optimization algorithms based on iterative water-filling for the MIMO IC (e.g. \cite{MIMOYe, CompetitiveMIMOScutari}) can be modeled as non-cooperative games, where each user is considered as a player that attempts to maximize its own utility selfishly. Such approaches may not converge in general or may converge to the Nash equilibrium (NE). It is well known that the Nash equilibrium is \emph{often} not Pareto-optimal \cite{NEvsPareto1}, since the best achievable performance characterized by Pareto boundary represents the set of optimal trade-offs among these conflicting/competing users' objectives. The trade-off of different users needs to be optimized by cooperative algorithms to achieve their joint outcome \cite{VS}.

A direct improvement from NE to Nash bargaining (NB) by cooperation for the MIMO IC has been studied in \cite{NBMIMOICChen}, where the case is studied in which the interference-plus-noise covariance matrix of each user approaches an identity matrix and the rate region becomes convex. A main branch of cooperative algorithms is the interference-pricing based method, where each user updates its own strategy to maximize its own utility minus the interference cost determined by the interference prices, which reflect the marginal change in utility per unit interference power.
This distributed interference-pricing algorithm has been used to solve (weighted) sum-rate maximization problems for the multi-user SISO and MISO IC \cite{PricingMonotonicCon}, multi-user single-stream MIMO IC \cite{PricingLocalIMIMO}, two-user MIMO IC \cite{InterferencePricingMIMO}.
A different pricing scheme is to balance the egoistic and altruistical strategies with different weights (i.e., prices), e.g., for the two-user MISO IC \cite{WalrasianEqRami, EduardMISOGame} and for the multi-user single-stream MIMO IC \cite{BalancingMIMO}. In \cite{BalancingMIMO}, the suboptimal maximum sum rate is achieved, although no convergence analysis is provided. However, most distributed cooperative algorithms for the MIMO IC (e.g., distributed pricing based algorithms \cite{InterferencePricingMIMO, PricingLocalIMIMO}) focus on maximizing (weighted) sum-utility rather than computing the whole utility region.

In fact, most approaches of parameterization, computation and weighted sum maximization are coordinated/cooperative algorithms, although they are not described in a game theoretic way.

For the MIMO IC, the achievable rate depends on both transmit and receive strategies involved in a more hard-coupling and complex expression than the MISO IC such that it is not straightforward to extend the \emph{implicit} or \emph{explicit} schemes achieving the complete Pareto boundary for the MISO IC to the MIMO IC. In order to illustrate the \emph{complete Pareto boundary}, we formulate a computation problem to maximize one rate while keeping the rate of the other users unchanged such that one rate can increase always along the same direction in the rate region until the boundary is reached. However in this computation problem, the hard-coupled beamformers exist not only in the objective but also in the constraints, which makes this beamformers optimization problems non-convex (even NP-hard). Therefore in this paper, as most current references on the MIMO IC, we focus on finding high quality sub-optimal operating points efficiently. Our main contributions are described as follows.

\begin{itemize}
     \item[\rmnum{1})] First, the two-user single-stream MIMO IC is firstly studied: a) We propose an equivalent form of the SINR expression based on the Hermitian angle (Proposition 1 in Section \Rmnum{2}-B), which gains more insight into the coupling of the transmit beamformers; b) We prove that the strict Pareto-optimal transmit power allocation policy is full power allocation at both the transmitters (Proposition 2 in Section \Rmnum{3}-A); c) The non-strict Pareto boundary, two ending points of strict Pareto boundary, and certain ZF points are computed exactly. (Section \Rmnum{3}-B).
     \item[\rmnum{2})] To compute the strict Pareto boundary of the two-user single-stream MIMO IC, we formulate a problem to maximize one rate while the other rate is fixed. This non-convex optimization problem is solved by the proposed alternating optimization algorithm \cite{NonlinearAlternatingOp} such that a convergent point is guaranteed to be achieved (Section \Rmnum{4} A-C).
     \item[\rmnum{3})] The proposed optimization algorithm can be extended to the multi-user scenario (Section \Rmnum{4}-D).
\end{itemize}

\textbf{Notation:}
$(\cdot)^*$, {$(\cdot)^T$}, $(\cdot)^H$, {$(\cdot)^\dagger$}, $\mathrm{Rank}(\cdot)$ and $\mathrm{Tr}(\cdot)$ denote complex conjugate, transpose, Hermitian and {Moore-Penrose pseudo inverse}, rank and trace, respectively.
$|\cdot|$, $\Re(\cdot)$, and $j$ denote the absolute value, the real part of a complex-valued number, and the imaginary unit, respectively.
$\lVert \mat{x}\lVert=\sqrt{\mat{x}^H\mat{x}}$ and $\overrightarrow {\mat{x}}=\frac{{\mat{x}}}{\lVert \mat{x}\lVert}$ denote the vector 2-norm and vector direction, respectively.
$\perp$, $\parallel$ and $\nparallel$ denote perpendicularity, parallelity and unparallelity, respectively.
$\mathcal{CN}(0, \mat{X})$ denotes a complex circularly-symmetric jointly-Gaussian probability density function with zero mean and covariance matrix $\mat{X}$.
{$\mat{X}\succeq \mat{0}$ or $\mat{X}\succ \mat{0}$ means $\mat{X}$ is a positive semidefinite matrix or a positive definite matrix.}
$\lambda_i(\mat{X})$ and $\mat{u}_i({\mat{X}})$ denote the $i$-th \emph{largest} eigenvalue of $\mat{X}$ and its corresponding eigenvector, respectively.
$\lambda_i(\mat{X}, \mat{Y})$ and $\mat{u}_i({\mat{X}}, \mat{Y})$ denote the $i$-th largest generalized eigenvalue of $\mat{X}, \mat{Y}$ and its corresponding eigenvector, respectively.
$\Pi_{\mat{X}}\stackrel{\Delta}{=}{\mat{X}}({\mat{X}^H}{\mat{X}})^{-1}{\mat{X}}^H$ denotes the orthogonal projection onto the column space of ${\mat{X}}$,
and $\Pi_{\mat{X}}^{\perp}\stackrel{\Delta}{=}\mat{I} - \Pi_{\mat{X}}$ denotes the orthogonal projection onto the orthogonal complement of the column space of ${\mat{X}}$.

\section{System Model}\label{sec:Signal Model}

\subsection{Signal Model}\label{sec:SignalModelA}

Consider a two-user MIMO IC denoted by $\mathrm{TX}_i \mapsto \mathrm{RX}_i, i=1,2$,
where each transmitter $\mathrm{TX}_i$ and each receiver $\mathrm{RX}_i$ are equipped with $N_T\geq2$ and $N_R\geq2$ antennas and a single data stream is transmitted in each user. In this IC, the received data at $\mathrm{RX}_i$ is modeled as
\begin{align}\label{eq:MIMOIF}
{y}_i\ = \mat{g}_i^H\left(\mat{H}_{ii}\mat{w}_i x_i + \mat{H}_{ki}\mat{w}_k x_k + \mat{n}_i \right),~i,~k\in\{1, 2\}, k\neq i \nonumber
\end{align}
where $x_i \sim \mathcal{CN}(0, 1)$ is the transmitted symbol of $\mathrm{TX}_i$ by the transmit beamformer $\mat{w}_i  \in {\mathbb{C}}^{{N_T}\times{1}}$. At $\mathrm{RX}_i$, $\mat{g}_i \in {\mathbb{C}}^{{N_R}\times{1}}$ is the receive beamformer, and $\mat{n}_i\in {\mathbb{C}}^{{N_R}\times{1}}\sim \mathcal{CN}(0, {\sigma}_i^2\mat{I})$ is the additive white Gaussian noise (AWGN) vector. The matrices $\mat{H}_{ii}, \mat{H}_{ki}\in {\mathbb{C}}^{{N_R}\times{N_T}}$ denote the flat fading channel-matrix of the direct link $\mathrm{TX}_i \mapsto \mathrm{RX}_i$ and the cross-talk link $\mathrm{TX}_k \mapsto \mathrm{RX}_i$, respectively. Each transmitter has a power constraint that we, without loss of generality, set to 1 and define the set of feasible transmit beamformers as $\mathcal{W}\stackrel{\Delta}{=}\left\{\mat{w}\in {\mathbb{C}}^{{N_T}\times{1}} : {\lVert \mat{w} \lVert}^2\leq1\right\}$.

\subsection{Rate with MMSE Receiver}\label{sec:RateMMSE}

Assume that the interference from the other transmitter is treated as additive Gaussian noise at each receiver. The achievable rate of the link $\mathrm{TX}_i \mapsto \mathrm{RX}_i$ is given by:
\begin{align}\label{eq:MMSER}
{R}_i(\mat{w}_1,\mat{w}_{2},\mat{g}_i) = \log_2\big(1+\mathrm{SINR}_i(\mat{w}_1,\mat{w}_{2},\mat{g}_i)\big),
\end{align}
where $\mathrm{SINR}_i(\mat{w}_1,\mat{w}_{2},\mat{g}_i) = \frac{|\mat{g}_i^H\mat{H}_{ii}\mat{w}_i|^2}{\sigma_i^2 + |\mat{g}_i^H\mat{H}_{ki}\mat{w}_k|^2}$. {In the linear transceiver design, it is well known that the MMSE filter is the optimal receiver for given transmit strategies.} In this paper, the MMSE filter $\mat{g}_i = \big({\sigma}_i^2\mat{I} +  \mat{H}_{ii}{\mat{w}_i}{\mat{w}_i}^H{\mat{H}_{ii}}^H + \mat{H}_{ki}{\mat{w}_k}{\mat{w}_k}^H{\mat{H}_{ki}}^H\big)^{-1}{\mat{H}_{ii}\mat{w}_i}$ is employed at $\mathrm{RX}_i$. Then, $\mathrm{SINR}_i(\mat{w}_1,\mat{w}_{2},\mat{g}_i)$ becomes
\begin{align}\label{eq:MMSESINR}
&\mathrm{SINR}_i(\mat{w}_1,\mat{w}_{2})  \nonumber \\
&=\mat{w}_i^H \underbrace{\mat{H}_{ii}^H\big({\sigma}_i^2\mat{I} +\mat{H}_{ki}{\mat{w}_k}{\mat{w}_k}^H{\mat{H}_{ki}}^H\big)^{-1}{\mat{H}_{ii}}}_{{\stackrel{\Delta}{=}\mat{A}_i(\mat{w}_k)}}\mat{w}_i.
\end{align}
The complex mathematical structure (inverse of the sum of matrices and product of matrices) causes a hard-coupling problem of $\mat{w}_1$ and $\mat{w}_2$ in the SINR in (\ref{eq:MMSESINR}), which makes it difficult to analyze the SINR directly. To gain an insight into this coupling problem, we propose an equivalent form of the SINR expression.
\begin{proposition}\label{pr:MMSESinrNew1}
For the two-user single-beam MIMO IC, the SINR in (\ref{eq:MMSESINR}) can be reformulated as
\begin{align}\label{eq:MMSESINRNew1}
&\mathrm{SINR}_i(\mat{w}_1,\mat{w}_{2}) = \nonumber \\
&\sin^2(\theta_{H,i})\frac{{\lVert \mat{H}_{ii}\mat{w}_i \lVert}^2}{\sigma_i^2} + \cos^2(\theta_{H,i}) \frac{{\lVert \mat{H}_{ii}\mat{w}_i \lVert}^2}{\sigma_i^2 + {\lVert \mat{H}_{ki}\mat{w}_k \lVert}^2},
\end{align}
where $\cos(\theta_{H,i})=\big |{\overrightarrow{\mat{H}_{ii}\mat{w}_i}}^H\cdot{\overrightarrow{\mat{H}_{ki}\mat{w}_k}}\big|$ and $\theta_{H,i}\in[0,\pi/2]$.
\end{proposition}

\begin{IEEEproof}
See Appendix A.
\end{IEEEproof}

Note that the $\mathrm{SINR}_i(\mat{w}_1,\mat{w}_{2})$ can be considered as a combination of $\frac{{\lVert \mat{H}_{ii}\mat{w}_i \lVert}^2}{\sigma_i^2}$ and $\frac{{\lVert \mat{H}_{ii}\mat{w}_i \lVert}^2}{\sigma_i^2 + {\lVert \mat{H}_{ki}\mat{w}_k \lVert}^2}$ with the weights $\sin^2(\theta_{H,i})$ and $\cos^2(\theta_{H,i})$. That is, $\mathrm{SINR}_i(\mat{w}_1,\mat{w}_{2})$ depends not only on the desired signal power ${\lVert \mat{H}_{ii}\mat{w}_i \lVert}^2$ and the interference power ${\lVert \mat{H}_{ki}\mat{w}_k \lVert}^2$, but also on the Hermitian angle $\theta_{H,i}$ between the directions ${\overrightarrow{\mat{H}_{ii}\mat{w}_i}}$ and ${\overrightarrow{\mat{H}_{ki}\mat{w}_k}}$. The SINR is coupled in a difficult way because of the existence of $\theta_{H,i}$. This is why it is more difficult to analyze the SINR of a MIMO IC than that of a MISO IC.

\section{Pareto Boundary and Computation of Some Key Points}\label{sec:Paretoboundary}

\subsection{Pareto Boundary}\label{sec:ParetoboundaryA}

The achievable rate region is defined as a set of the achievable rate pairs with all the feasible beamformers
\begin{align}\label{eq:RateRegion}
\mathcal{R} &\stackrel{\Delta}{=}  \bigcup_{\mat{w}_1,\mat{w}_2\in\mathcal{W}} \left( {R}_1(\mat{w}_1,\mat{w}_{2}),{R}_2(\mat{w}_1,\mat{w}_{2}) \right). \nonumber
\end{align}
Note that the achievable rate region $\mathcal{R}$ is not the capacity region.
Its outermost boundary is called Pareto boundary in this paper, which can be denoted by a set $\mathcal{R}^\star\stackrel{\Delta}{=}\bigcup(R_1^\star,R_2^\star)$ where $(R_1^\star,R_2^\star)$ is a Pareto-optimal point. More precisely, Pareto-optimality is defined as follows.

\begin{definition}
A rate pair $(R_1^\star, R_2^\star) \in \mathcal{R}$ is [strict] Pareto-optimal iff there does not exist another rate pair $(R_1, R_2)>(R_1^\star, R_2^\star)$ [$(R_1, R_2)\geq(R_1^\star, R_2^\star)$ and $(R_1, R_2)\neq(R_1^\star, R_2^\star)$] with $(R_1, R_2) \in \mathcal{R}$, where the inequality is component-wise.
\end{definition}
\begin{figure}
\begin{center}
\includegraphics[scale=0.72]{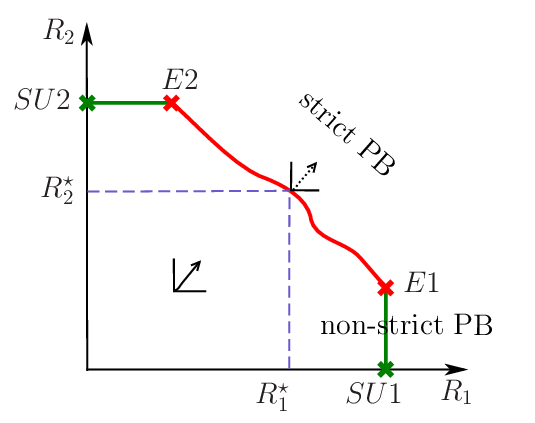}
\caption{Achievable rate region and Pareto boundary in two-user case}\label{fig:no1}
\end{center}
\end{figure}
As shown in Fig. \ref{fig:no1}, the Pareto boundary consists of the \emph{strict Pareto boundary} (the upperright
part graphically, denoted by "strict PB") and the \emph{non-strict Pareto boundary} (including the vertical
part and the horizontal part graphically, denoted by "non-strict PB"), divided by "E1" and "E2". "E1"
and "E2", "SU1" and "SU2" mean two ending points of the strict Pareto boundary and two single user
points, which will be studied in Section \Rmnum{3}-B. More precisely, for an arbitrary point on the strict Pareto
boundary, it is impossible to improve one rate without simultaneously decreasing the other. For a point
on the non-strict Pareto boundary, one rate can be further
improved while the other rate remains the maximum rate. In particular, the strict Pareto boundary can
be characterized as follows.

\begin{proposition}\label{pr:ParetoBoundaryP}
For the two-user single-stream MIMO IC, all the operating points on the strict Pareto boundary can be achieved only when both the transmitters spend the full power, i.e., $\lVert \mat{w}_1 \lVert^2=\lVert \mat{w}_2 \lVert^2=1$.
\end{proposition}
\begin{IEEEproof}
See Appendix B.
\end{IEEEproof}
\begin{remark}
In fact, Proposition 2 has solved a strict Pareto-optimal transmit power allocation problem
in this scenario. When the strict Pareto-optimal power allocation policy is employed, i.e., both the
transmitters spend the full power, the two strict Pareto-optimal transmit beamformers design reduces
to the optimization of two transmit beamforming patterns.
\end{remark}
Here, we define a set of all the beamformers with full transmit power as $\mathcal{W_{FP}}\stackrel{\Delta}{=}\left\{\mat{w}\in {\mathbb{C}}^{{{N_T}}\times{1}} : {\lVert \mat{w} \lVert}^2=1\right\}$
Note that all the strict Pareto-optimal transmit beamformers should be in the set $\mathcal{W_{FP}}$.

\subsection{Computation of Some Key Points}\label{sec:ParetoboundaryB}

In this part, we compute exactly the non-strict Pareto boundary, two ending points of the strict Pareto boundary and certain ZF operating points.

\subsubsection{\textbf{Single-User Points} $SU1(\overline{R}_1,0)$ and $SU2(0,\overline{R}_2)$}
A single-user point can be easily achieved when only one $\mathrm{TX}_i$ works and simultaneously operates "egoistically" to maximize its own rate. The maximum achievable rate $\overline{R}_i$ of the link $\mathrm{TX}_i\mapsto \mathrm{RX}_i$ and its associated "egoistic" strategy $\mat{w}_i^{Ego}$ are
\begin{align}\label{eq:SURate}
\overline{R}_i=\log_2\Big(1+\frac{\lambda_{1}(\mat{H}_{ii}^H\mat{H}_{ii})}{\sigma_i^2}\Big), ~\mat{w}_i^{Ego}=\mat{u}_{1}(\mat{H}_{ii}^H\mat{H}_{ii})~\forall i.
\end{align}

\subsubsection{\textbf{Ending Points of Strict Pareto Boundary} $E1(\overline{R}_1, \underline{R}_2)$ and $E2(\underline{R}_1, \overline{R}_2)$}
Each ending point of the strict Pareto boundary can be achieved when one transmitter employs an "altruistic" strategy to create no
interference to the other receiver and simultaneously to maximize its own rate and the other transmitter
operates "egoistically". For $E1(\overline{R}_1, \underline{R}_2)$, we easily find from (\ref{eq:MMSESINRNew1}) that $\theta_{H,1}=\pi/2$ results in no interference in the cross-talk link $\mathrm{TX}_2 \mapsto \mathrm{RX}_1$. How to find the "altruistic" strategy $\mat{w}_2^{Alt}$ is shown as follows.
\begin{proposition}\label{pr:TP}
$E1(\overline{R}_1,\underline{R}_2)$ can be achieved by $(\mat{w}_1^{Ego}, \mat{w}_2^{Alt})$, where $\overline{R}_1$ and $\mat{w}_1^{Ego}$ are in (\ref{eq:SURate}) and
\begin{align}\label{eq:OptiBeamTP1}
&\underline{R}_2 = \log_2{\left(1+\mat{w}_2^{Alt,H}\mat{A}_2(\mat{w}_1^{Ego})\mat{w}_2^{Alt}\right)}, \nonumber \\
&\mat{w}_2^{Alt} = \overrightarrow{\Pi_{\mat{H}_{21}^H\mat{H}_{11}\mat{w}_1^{Ego}}^{\perp}\mat{u}_1\left(\mat{B}_{1}, \Pi_{\mat{H}_{21}^H\mat{H}_{11}\mat{w}_1^{Ego}}^{\perp}\right)},
\end{align}
with $\mat{B}_{1} \stackrel{\Delta}{=} \Pi_{\mat{H}_{21}^H\mat{H}_{11}\mat{w}_1^{Ego}}^{\perp} \mat{A}_2(\mat{w}_1^{Ego})\Pi_{\mat{H}_{21}^H\mat{H}_{11}\mat{w}_1^{Ego}}^{\perp}$.
\end{proposition}
\begin{IEEEproof}
See Appendix C.
\end{IEEEproof}

Similarly, $E2(\underline{R}_1, \overline{R}_2)$ with $(\mat{w}_1^{Alt}, \mat{w}_2^{Ego})$ can be easily obtained by interchanging the indices.

\subsubsection{\textbf{Non-strict Pareto-Optimal Points} $(R_1^{\star},R_2^{\star})$}
For the non-strict Pareto boundary, either the horizontal part or the vertical part starts and ends with a single user point and an ending point. Therefore, an arbitrary point $(R_1^\star, R_2^\star)$ on the non-strict Pareto boundary can be computed as
\begin{align}\label{eq:WeakPBRate}
R_i^\star = \gamma \cdot\underline{R}_i~~\mathrm{and}~~
R_k^\star = \overline{R}_k, ~~\forall i, ~k\in\{1,2\},~k\neq i \nonumber
\end{align}
where $i=1$ and $i=2$ correspond to the horizontal part and the vertical part, respectively. The scalar $\gamma$ satisfies $\gamma \in [0,1)$. The point $(R_1^\star, R_2^\star)$ becomes a single-user point or an ending point when $\gamma = 0$ or $\gamma = 1$, respectively. The associated non-strict Pareto-optimal
transmit strategies are
\begin{align}\label{eq:WeakPBBeam}
\mat{w}_i^\star =  \sqrt{\gamma}\cdot\mat{w}_i^{Alt} ~~\mathrm{and}~~
\mat{w}_k^{\star} = \mat{w}_k^{Ego}, \nonumber
\end{align}
from which we find that it is not necessary for both the transmitters to spend full power simultaneously
to achieve the non-strict Pareto boundary. Thus, the non-strict Pareto-optimal power allocation policy is
different from the strict Pareto-optimal power allocation policy (Proposition 2).

\subsubsection{\textbf{Zero-Forcing (ZF) Points} $ZF(R_1^{ZF}, R_2^{ZF})$} ZF points are achieved when there is no interference between different users. {Although these points are not on the Pareto boundary, it is still interesting to study ZF strategies if there exsits an additional requirement (like interference temperature or secrecy constraints) that each transmitter does not leak its own signal to other receivers.}

In (\ref{eq:MMSESINRNew1}), we find that $\theta_{H,1}=\theta_{H,2}={\pi}/{2}$ results in no interference in the cross-talk links $\mathrm{TX}_2 \mapsto \mathrm{RX}_1$ and $\mathrm{TX}_1 \mapsto \mathrm{RX}_2$ simultaneously. The ZF conditions are
\begin{align}
\theta_{H,1}&={\pi}/{2}\Leftrightarrow{\overrightarrow{\mat{H}_{11}\mat{w}_1}}^H\cdot{\overrightarrow{\mat{H}_{21}\mat{w}_2}}=0 \nonumber \\
&\Leftrightarrow{\mat{H}_{11}\mat{w}_1}\perp{\mat{H}_{21}\mat{w}_2} \Leftrightarrow {\mat{w}_2}\perp{\mat{H}_{21}^H\mat{H}_{11}\mat{w}_1},
\nonumber \\
\theta_{H,2}&={\pi}/{2}\Leftrightarrow{\overrightarrow{\mat{H}_{22}\mat{w}_2}}^H\cdot{\overrightarrow{\mat{H}_{12}\mat{w}_1}}=0  \nonumber \\
&\Leftrightarrow{\mat{H}_{22}\mat{w}_2}\perp{\mat{H}_{12}\mat{w}_1} \Leftrightarrow {\mat{w}_2}\perp{\mat{H}_{22}^H\mat{H}_{12}\mat{w}_1},
\end{align}
{from which and under a sufficient (and necessary only when $N_T = N_R = 2$) condition\footnote{This condition is the same as that in \cite{Two-cellMIMO}, while we derive it from a different perspective (Hermitian angle in (\ref{eq:MMSESINRNew1})).}, i.e., ${\mat{H}_{21}^H\mat{H}_{11}\mat{w}_1}\parallel{\mat{H}_{22}^H\mat{H}_{12}\mat{w}_1}$, we obtain some of the ZF transmit strategies
\begin{align}\label{eq:BeamZF}
&\mat{w}_1^{ZF} = \mat{u}_\ell(\mat{H}_{22}^H\mat{H}_{12}, \mat{H}_{21}^H\mat{H}_{11}), \forall \ell \in \{1,2,...,N_T\}, \nonumber \\
&\mat{w}_2^{ZF} = \sum_{\ell=1}^{N_T-1}{c_\ell\mat{u}_\ell(\Pi_{\mat{H}_{21}^H\mat{H}_{11}\mat{w}_1^{ZF}}^{\perp})},
\end{align}
}
where $\{c_\ell\}_{\ell=1}^{N_T-1}$ are complex-valued numbers and satisfy $\sum_{\ell=1}^{N_T-1}|c_\ell|^2=1$.

The $ZF(R_1^{ZF}, R_2^{ZF})$ can be achieved by $(\mat{w}_1^{ZF},\mat{w}_{2}^{ZF})$ as
\begin{align}\label{eq:RateZF}
{R}_i^{ZF}(\mat{w}_i^{ZF},\mat{w}_{2}^{ZF}) = \log_2\left(1+\frac{{\lVert \mat{H}_{ii}\mat{w}_i^{ZF} \lVert}^2}{\sigma_i^2}\right)~\forall i.
\end{align}

\section{Computation of the Strict Pareto Boundary}\label{sec:ComputePB}

Since the rate region of the two-user single-beam MIMO IC is always a normal region{\footnote{A set $\mathcal{G} \subseteq \mathbb{R}_n^+$ is called a normal region if for any two points $\mat{x} \in \mathcal{G}, \mat{x}' \in \mathbb{R}_n^+$ such that if $\mat{x}' \leq \mat{x}$, then $\mat{x}' \in \mathcal{G}$, too.}} according to Proposition 2, there exists only one intersection point between the line $R_i(\mat{w}_1, \mat{w}_2) = R_i^\star$ where $R_i^\star \in (\underline{R}_i, \overline{R}_i)$ and the strict Pareto boundary. Thus, an arbitrary point on the strict Pareto boundary can be uniquely determined when one rate is fixed and the other rate is maximized. This motivates us to propose the following optimization problem
\begin{equation*}\label{eq:P0}
\mathrm{(P0)}
\left\{
\begin{aligned}
\max_{\mat{w}_1, \mat{w}_2\in\mathcal{W_{FP}}} ~~&{\mathrm{SINR}_1(\mat{w}_1, \mat{w}_2)} \nonumber \\
\mathrm{s.t.}~~~~~~~&{\mathrm{SINR}_2(\mat{w}_1, \mat{w}_2)} =  \mathrm{SINR}_2^\star.
\end{aligned}
\right.
\end{equation*}
where $\mathrm{SINR}_2^\star\in(2^{\underline{R}_2}-1, 2^{\overline{R}_2}-1)$ is a SINR constraint, and $\mat{w}_1, \mat{w}_2$ should be in $\mathcal{W_{FP}}$ according to Proposition 2. Then, $(R_1^\star, R_2^\star)=\left(\log_2\left(1+\mathrm{SINR}_1(\mat{w}_1^\star, \mat{w}_2^\star)\right), \log_2\left(1+\mathrm{SINR}_2^\star(\mat{w}_1^\star, \mat{w}_2^\star)\right)\right)$ is achieved by the optimal solution $(\mat{w}_1^\star, \mat{w}_2^\star)$ to $\mathrm{(P0)}$.

For $\mathrm{(P0)}$, direct joint optimization of $\mat{w}_1$ and $\mat{w}_2$ is analytically intractable due to the hard-coupling problem of them in both the objective and the constraints. To solve $\mathrm{(P0)}$, an alternating optimization algorithm \cite{NonlinearAlternatingOp} is applied to optimize $\mat{w}_1$ and $\mat{w}_2$ alternatively by solving two single-beamformer optimization problems at each iteration. In the following, how to solve each single-beamformer problem is studied.

\subsection{{Optimization of $\mat{w}_1$}}\label{sec:OpTX1}

For a given \emph{feasible} $\mat{w}_2$ (the feasibility of $\mat{w}_2$ will be studied in Proposition 4), the problem $\mathrm{(P0)}$ becomes a single-beamformer optimization problem w.r.t. $\mat{w}_1$. Its constraint is
\begin{subequations}\label{eq:FractionConstriant}
\begin{align}
\mat{w}_2^H&\mat{H}_{22}^H\big({\sigma}_2^2\mat{I} +\mat{H}_{12}{\mat{w}_1}{\mat{w}_1}^H{\mat{H}_{12}}^H\big)^{-1}{\mat{H}_{22}\mat{w}_2}=\mathrm{SINR}_2^\star \nonumber \\
\stackrel{(a)}{\Leftrightarrow}
&\mat{w}_2^{H}\mat{H}_{22}^H\mat{H}_{22}\mat{w}_2 - \frac{|\mat{w}_2^{H}\mat{H}_{22}^H\mat{H}_{12}{\mat{w}_1}|^2}{{\sigma}_2^2+\lVert \mat{H}_{12}{\mat{w}_1} \lVert^2} = \sigma_2^2\mathrm{SINR}_2^\star,\nonumber\\
\stackrel{(b)}{\Leftrightarrow}
&\frac{{\mat{w}_1}^H\mat{H}_{12}^H\mat{H}_{22}\mat{w}_2\mat{w}_2^{H}\mat{H}_{22}^H\mat{H}_{12}\mat{w}_1}{\mat{w}_1^H({\sigma}_2^2\mat{I}+\mat{H}_{12}^H\mat{H}_{12}){\mat{w}_1}} \nonumber \\
&~~~~~~~~~~~~~~~~~~~= \mat{w}_2^H\mat{H}_{22}^H\mat{H}_{22}\mat{w}_2 - \sigma_2^2\mathrm{SINR}_2^\star,\\
\stackrel{(c)}{\Leftrightarrow}
&\mat{w}_1^H\mat{C}(\mat{w}_2)\mat{w}_1 = 0~\mathrm{and}~\mat{w}_2^H\mat{H}_{22}^H\mat{H}_{22}\mat{w}_2 \geq \sigma_2^2\mathrm{SINR}_2^\star.
\end{align}
\end{subequations}
The transformation $(a)$ is based on the matrix inverse lemma. The transformation $(b)$ is due to $\lVert \mat{w}_1 \lVert^2 =1$. In the transformation $(c)$, the nonnegative left-hand side of (\ref{eq:FractionConstriant}a) demands $\mat{w}_2^H\mat{H}_{22}^H\mat{H}_{22}\mat{w}_2 \geq \sigma_2^2\mathrm{SINR}_2^\star$, and $\mat{C}$ is a Hermitian matrix defined as
\begin{align}\label{eq:WfMatrix}
&\mat{C}(\mat{w}_2)\stackrel{\Delta}{=}\mat{H}_{12}^H\mat{H}_{22}\mat{w}_2\mat{w}_2^{H}\mat{H}_{22}^H\mat{H}_{12} \nonumber \\
&-(\mat{w}_2^H\mat{H}_{22}^H\mat{H}_{22}\mat{w}_2 - \sigma_2^2\mathrm{SINR}_2^\star)\cdot({\sigma}_2^2\mat{I}+\mat{H}_{12}^H\mat{H}_{12}).
\end{align}

Then, $\mat{w}_1$ can be optimized by
\begin{equation*}\label{eq:P1}
\mathrm{(P1)}
\left\{
\begin{aligned}
\max_{\mat{w}_1\in\mathcal{W_{FP}}} ~~&{\mat{w}_1^H \mat{A}_1(\mat{w}_2)\mat{w}_1}  \\
\mathrm{s.t.}~~~~&{{\mat{w}_1}^H\mat{C}(\mat{w}_2)\mat{w}_1} =  0 \\
\end{aligned}
\right.
\end{equation*}
where $\mat{C}(\mat{w}_2)$ and $\mat{A}_1(\mat{w}_2)$ are Hermitian matrices. Observe that the problem $\mathrm{(P1)}$ is a homogeneous quadratically constrained quadratic program (QCQP) and the objective function is convex but the convexity of constraints is unclear. Generally, it is difficult to solve this non-convex problem.

Note that $\mat{w}_1^H\mat{X}{\mat{w}_1} = \mathrm{Tr}(\mat{X}\mat{W}_1)$ for any matrix $\mat{X}$, where $\mat{W}_1 = \mat{w}_1\mat{w}_1^H$ is a rank-one Hermitian positive semidefinite matrix. By the semidefinite programming and rank relaxation (SDR) method, $\mathrm{(P1)}$ can be transformed to
\begin{equation*}\label{eq:P2}
\mathrm{(P2)}
\left\{
\begin{aligned}
\max_{\mat{W}_1\succeq\mat{0}} ~~&\mathrm{Tr}\left(\mat{A}_1(\mat{w}_2)\mat{W}_1\right) \\
\mathrm{s.t.}~~~~&\mathrm{Tr}\left(\mat{C}(\mat{w}_2)\mat{W}_1\right)  = 0 \\
&\mathrm{Tr}\left(\mat{W}_1\right) = 1.
\end{aligned}
\right.
\end{equation*}
Observe that the SDR $\mathrm{(P2)}$ is convex and solvable, i.e., its respective finite optimal solutions exist for a feasible $\mat{w_2}$, based on Weierstrass' Theorem. Its optimal solution $\mat{W}_1^\star$ is efficiently obtained by a convex optimization toolbox, e.g., SeDuMi \cite{SeDuMi} or CVX \cite{CVXTool}.
However, the rank of $\mat{W}_1^\star$ to $\mathrm{(P2)}$ is usually more than one because we have discarded the rank constraint $\mathrm{Rank}(\mat{W}_1)=1$. Therefore, we need to extract an optimal rank-one solution $\mat{w}_1$ to $\mathrm{(P1)}$ from $\mat{W}_1^\star$. If $\mathrm{Rank}(\mat{W}_1^\star) = 1$, it is clear $\mat{w}_1 = \mat{u}_1(\mat{W}_1^\star)$. Otherwise, other tight matrix rank-one decomposition methods are needed. In \cite{SDPHuang2011}, Ai et al. have proven a matrix rank-one decomposition theorem and used it to show that \emph{the SDRs of a large class of complex-valued homogeneous QCQPs with not more than 4 constraints are in fact tight}\footnote{Note that the application of Theorem 2.2 and Theorem 2.3 in \cite{SDPHuang2011} needs $N_T\geq3$. It implies that $\mathrm{TX}_k~ \forall k$ should have $N_T\geq3$ antennas in our scenario.}. Since the problem $\mathrm{(P1)}$ as a homogeneous QCQP with 2 constraints, an optimal $\mat{w}_1$ to the QCQP $\mathrm{(P1)}$ can be reconstructed from $\mat{W}_1^\star$ to the SDR $\mathrm{(P2)}$ based on the theorem and algorithm of the matrix rank-one decomposition in \cite{SDPHuang2011}.

\begin{remark}
In $\mathrm{(P1)}$, if $\mat{C}(\mat{w}_2)$ is a positive/negative semidefinite matrix without full rank, $\mat{w}_1$ should and must be in the null space of $\mat{C}(\mat{w}_2)$ to satisfy $\mat{w}_1^H\mat{C}(\mat{w}_2)\mat{w}_1=0$. According to the proof in Appendix C, $\mat{w}_1$ can be expressed by $\overrightarrow{\mat{U}_1\mat{U}_1^H\mat{p}_{1}}$ where $\mat{U}_1\in \mathbb{C}^{N_T \times (N_T-\mathrm{Rank}({\mat{C}(\mat{w}_2)}))}$ consists of $N_T-\mathrm{Rank}({\mat{C}(\mat{w}_2)})$ eigenvectors corresponding to zero eigenvalues of ${\mat{C}(\mat{w}_2)}$. Then, $\mathrm{(P1)}$ is equivalent to
\begin{align}\label{eq:OptiP1}
&\max_{\mat{p}_{1}\in\mathbb{C}^{N_T\times 1}} {\frac{\mat{p}_{1}^H { \mat{U}_1\mat{U}_1^H \mat{A}_1(\mat{w}_2) \mat{U}_1\mat{U}_1^H } \mat{p}_{1}} {\mat{p}_{1}^H \mat{U}_1\mat{U}_1^H\mat{p}_{1}}},
\end{align}
from which it is easy to derive the optimal solution $\mat{p}_{1}^{opt} = \mat{u}_1\big({ \mat{U}_1\mat{U}_1^H \mat{A}_1(\mat{w}_2) \mat{U}_1\mat{U}_1^H},\mat{U}_1\mat{U}_1^H\big)$. Therefore, the optimal solution to $\mathrm{(P1)}$ is $\mat{w}_1=\overrightarrow{\mat{U}_1\mat{U}_1^H\mat{p}_{1}^{opt}}$.
\end{remark}

\subsection{{Optimization of $\mat{w}_2$}}\label{sec:OpTX2}

For a given \emph{feasible} $\mat{w}_1$, $\mathrm{(P0)}$ becomes another single-beamformer optimization problem w.r.t. $\mat{w}_2$. Maximization of its objective function is
\begin{equation}\label{eq:P0ObjectiveF}
\begin{aligned}
\max_{\mat{w}_2\in\mathcal{W_{FP}}} &\mat{w}_1^H\mat{H}_{11}^H\left(\sigma_1^2\mat{I} + \mat{H}_{21}\mat{w}_{2}\mat{w}_{2}^H\mat{H}_{21}^H\right)^{-1}\mat{H}_{11}\mat{w}_1 \\
\Longleftrightarrow
\max_{\mat{w}_2\in\mathcal{W_{FP}}} &\mat{w}_1^{H}\mat{H}_{11}^H\mat{H}_{11}\mat{w}_1 - \frac{|\mat{w}_1^{H}\mat{H}_{11}^H\mat{H}_{21}{\mat{w}_2}|^2}{{\sigma}_1^2+{\mat{w}_2}^H\mat{H}_{21}^H\mat{H}_{21}{\mat{w}_2}}\\
\Longleftrightarrow
\min_{\mat{w}_2\in\mathcal{W_{FP}}} &\frac{|\mat{w}_1^{H}\mat{H}_{11}^H\mat{H}_{21}{\mat{w}_2}|^2}{{\sigma}_1^2+{\mat{w}_2}^H\mat{H}_{21}^H\mat{H}_{21}{\mat{w}_2}} \\
\Longleftrightarrow
\min_{\mat{w}_2\in\mathcal{W_{FP}}} &\frac{{\mat{w}_2}^H \mat{C}_{1}(\mat{w}_1) \mat{w}_2}{\mat{w}_2^H \mat{C}_{2} {\mat{w}_2}}
\end{aligned}
\end{equation}
where $\mat{C}_{1}(\mat{w}_1) \stackrel{\Delta}{=} \mat{H}_{21}^H\mat{H}_{11}\mat{w}_1\mat{w}_1^{H}\mat{H}_{11}^H\mat{H}_{21}$ and $\mat{C}_{2} \stackrel{\Delta}{=} {\sigma}_1^2\mat{I}+\mat{H}_{21}^H\mat{H}_{21}$ are Hermitian matrices.

Then, $\mat{w}_2$ can be optimized by
\begin{equation*}\label{eq:P3}
\mathrm{(P3)}
\left\{
\begin{aligned}
\min_{\mat{w}_2\in\mathcal{W_{FP}}} ~~&\frac{{\mat{w}_2}^H \mat{C}_{1}(\mat{w}_1) \mat{w}_2}{\mat{w}_2^H \mat{C}_{2} {\mat{w}_2}} \\
\mathrm{s.t.}~~~~&{\mat{w}_2^H \mat{A}_2(\mat{w}_1)\mat{w}_2}=\mathrm{SINR}_2^\star.
\end{aligned}
\right.
\end{equation*}
Observe that $\mathrm{(P3)}$ is a fractional QCQP problem. The objective function is not even a quasi-convex function due to the convexity of both the nominator function and denominator function. To deal with this problem, we transform it by the SDR to
\begin{equation*}\label{eq:P4}
\mathrm{(P4)}
\left\{
\begin{aligned}
\min_{\mat{W}_2\succeq\mat{0}} ~~&\frac{\mathrm{Tr}\big(\mat{C}_{1}(\mat{w}_1)\mat{W}_2\big)}{\mathrm{Tr}\big(\mat{C}_{2}\mat{W}_2\big)} \\
\mathrm{s.t.}~~~~&\mathrm{Tr}\big(\mat{A}_2(\mat{w}_1)\mat{W}_2\big)  = \mathrm{SINR}_2^\star \\
&\mathrm{Tr}\big(\mat{W}_2\big) = 1,
\end{aligned}
\right.
\end{equation*}
which is still a non-convex problem. Fortunately, the fractional structure can be removed by a variation of the Charnes-Cooper variable transformation \cite{Charnes-Cooper}.
Define the transformed variable $\mat{Q}=s{\mat{W}_2}$ with {$s =\frac{1}{\mathrm{Tr}\big(\mat{C}_{2}\mat{W}_2\big)}$}. Then, $\mathrm{(P4)}$ becomes
\begin{equation*}\label{eq:P5}
\mathrm{(P5)}
\left\{
\begin{aligned}
\min_{\mat{Q},~s} ~~&{\mathrm{Tr}\big(\mat{C}_{1}(\mat{w}_1)\mat{Q}\big)} \\
\mathrm{s.t.}~~~&\mathrm{Tr}\big(\mat{A}_2(\mat{w}_1)\mat{Q}\big)  = s\cdot\mathrm{SINR}_2^\star \\
&{\mathrm{Tr}\big(\mat{C}_{2}\mat{Q}\big)} = 1, ~\mathrm{Tr}\big(\mat{Q}\big) = s \\
&\mat{Q}\succeq\mat{0},~\frac{1}{\lambda_1(\mat{C}_{2})}\leq s\leq \frac{1}{\lambda_N(\mat{C}_{2})}.
\end{aligned}
\right.
\end{equation*}
which is a convex problem w.r.t. $\mat{Q}$ and $s$ and solvable (see the Appendix D). By a convex optimization toolbox, we can obtain the optimal solution $(\mat{Q}^\star, s^\star)$. Then, the optimal solution to $\mathrm{(P4)}$ can be easily obtained by $\mat{W}_2^\star=\frac{\mat{Q}^\star}{s^\star}$.
Observe that $\mathrm{(P3)}$ is equivalent to a homogeneous QCQP with 3 constraints. Therefore, by the matrix rank-one decomposition method, an optimal rank-one solution $\mat{w}_2$ to $\mathrm{(P3)}$ can be extracted from $\mat{W}_2^\star$ when $\mathrm{Rank}(\mat{W}_2^\star)>1$.

\begin{remark}
In $\mathrm{(P3)}$, if $\mat{D} \stackrel{\Delta}{=} \mat{A}_2(\mat{w}_1)-\mathrm{SINR}_2^\star \cdot\mat{I}$ is a positive/negative semidefinite matrix without full rank, $\mat{w}_1$ should and must be in the null space of $\mat{D}$ to satisfy $\mat{w}_2^H\mat{D}\mat{w}_2=0$. According to the proof in Appendix C, $\mat{w}_2$ can be expressed by $\overrightarrow{\mat{U}_2\mat{U}_2^H\mat{p}_{2}}$ where $\mat{U}_2\in \mathbb{C}^{N_T \times (N_T-\mathrm{Rank}(\mat{D}))}$ consists of $N_T-\mathrm{Rank}(\mat{D})$ eigenvectors corresponding to zero eigenvalues of $\mat{D}$. Then, $\mathrm{(P3)}$ is equivalent to
\begin{align}\label{eq:OptiP2}
&\max_{\mat{p}_{2}\in\mathbb{C}^{N_T\times 1}} \frac{\mat{p}_{2}^H { \mat{U}_2\mat{U}_2^H \mat{C}_1(\mat{w}_1) \mat{U}_2\mat{U}_2^H } \mat{p}_{2}}{\mat{p}_{2}^H { \mat{U}_2\mat{U}_2^H \mat{C}_2 \mat{U}_2\mat{U}_2^H } \mat{p}_{2}},
\end{align}
from which it is easy to derive the optimal solution $\mat{p}_{2}^{opt} = \mat{u}_1\big({ \mat{U}_2\mat{U}_2^H \mat{C}_1(\mat{w}_1) \mat{U}_2\mat{U}_2^H}, \mat{U}_2\mat{U}_2^H \mat{C}_2 \mat{U}_2\mat{U}_2^H\big)$. Therefore, the optimal solution to $\mathrm{(P3)}$ is $\mat{w}_2=\overrightarrow{\mat{U}_2\mat{U}_2^H\mat{p}_{2}^{opt}}$.
\end{remark}

\subsection{{Algorithm}}\label{sec:IAA}

In this section, algorithm discussions are given to gain some insights into the proposed alternating optimization algorithm.

\subsubsection{\textbf{Algorithm Description}}\label{sec:Algorithm}

A feasible initial $\mat{w}_2$ for optimization of $\mathrm{TX}_1$ can be obtained as follows.
\begin{proposition}\label{pr:Wf}
For a given $\mathrm{SINR}_2^\star\in(2^{\underline{R}_2}-1, 2^{\overline{R}_2}-1)$, $(\mat{w}_1, \mat{w}_2)$ is always a feasible solution pair to $\mathrm{(P0)}$ if $\mat{w}_1\in\mathcal{W_{FP}}$, $\mat{w}_2\in\mathcal{W_{F}}$ with
\begin{align}
\mathcal{W_{F}}\stackrel{\Delta}{=} \Big\{&\mat{w}_2 \in \mathcal{W_{FP}}: \mat{w}_2^{H}\mat{H}_{22}^H\mat{H}_{22}\mat{w}_2 \geq \sigma_2^2\mathrm{SINR}_2^\star, \nonumber \\
&~\lambda_1\big(\mat{C}(\mat{w}_2)\big)\cdot \lambda_{N_T}\big(\mat{C}(\mat{w}_2)\big) \leq 0\Big\},
\end{align}
where $\mat{C}(\mat{w}_2)$ is defined in (\ref{eq:WfMatrix}).
\end{proposition}
\begin{IEEEproof}
See Appendix E.
\end{IEEEproof}
That is, $\mathrm{(P1)}$ equivalent to $\mathrm{(P0)}$ with a fixed ${\mat{w}}_2 \in \mathcal{W_{F}}$ always has at least one feasible point ${\mat{w}}_1$ in $\mathcal{W_{FP}}$ (more analysis of initialization will be given in Section \Rmnum{4}-C-2).

The proposed alternating optimization algorithm with any initial $\mat{w}_2 \in \mathcal{W_{F}}$ is described in pseudo-code as Algorithm \ref{alg:A1}:

\begin{algorithm}
\DontPrintSemicolon
\KwIn{$\mat{w}_2^{Alt}$, $\mat{w}_2^{Ego}$, an arbitrary $R_2^\star \in \big({\underline{R}_2}, {\overline{R}_2} \big)$, and $\mathrm{SINR}_2^\star = 2^{R_2^\star}-1$.}
\KwOut{ A convergent point $(R_1^{(\ell)}, R_2^\star)$ by $(\mat{w}_1^{(\ell)}, \mat{w}_2^{(\ell)})$.}
\Begin{
Initialization:\;
Set a feasible $\mat{w}_2^{(0)}\in\mathcal{W_F}$, $\ell=0$.

\While{some convergence criterion is not satisfied}{ $\ell++$.\;

Given $\mat{w}_2^{(\ell-1)}$, obtain an optimal $\mat{W}_1$ to $\mathrm{(P2)}$.\;

Extract an optimal $\mat{w}_1^{(\ell)}$ from $\mat{W}_1$ to $\mathrm{(P1)}$; \;

Given $\mat{w}_1^{(\ell)}$, obtain an optimal $(\mat{Q}, s)$ to $\mathrm{(P5)}$ and an optimal $\mat{W}_2 = \frac{\mat{Q}}{s}$ to $\mathrm{(P4)}$.\;

Extract an optimal $\mat{w}_2^{(\ell)}$ from $\mat{W}_2$ to $\mathrm{(P3)}$.\;

Compute $R_1^{(\ell)} = \log_2\big(1 + \mat{w}_1^{(\ell),H}\mat{A}_1(\mat{w}_2^{(\ell)})\mat{w}_1^{(\ell)}\big)$. \;
}
}
\caption{Two-User Alternating Optimization Algorithm}\label{alg:A1}
\end{algorithm}

\subsubsection{\textbf{Algorithm Analysis}}\label{sec:convergence}

In this section, we discuss the proposed algorithm in the following aspects: \rmnum{1}) the convergence, \rmnum{2}) the quality of the solution, \rmnum{3}) the implementation, \rmnum{4}) the complexity.

\emph{\rmnum{1}) Convergence Analysis:}
Based on the results in Section \Rmnum{4}-A and Section \Rmnum{4}-B, a global optimal solution to each single-beamformer optimization problem $\mathrm{(P1)}$ and $\mathrm{(P3)}$ can be obtained at each iteration. We will show that the sequence $\left\{\mathrm{SINR}_1(\mat{w}_1^{(\ell)}, \mat{w}_2^{(\ell)})\right\}_{\ell=1}^{\infty}$ by Algorithm \ref{alg:A1} monotonically increases and converges, i.e., $\mathrm{SINR}_1(\mat{w}_1^{(\ell+1)}, \mat{w}_2^{(\ell+1)}) \geq \mathrm{SINR}_1(\mat{w}_1^{(\ell)}, \mat{w}_2^{(\ell)})~\forall \ell$.

Denote the optimization of $\mat{w}_1$ and the optimization of $\mat{w}_2$ by the mapping functions $\mat{w}_1 = \Phi(\mat{w}_2)$ and $\mat{w}_2 = \Theta(\mat{w}_1)$, respectively. Then, the procedure of Algorithm \ref{alg:A1} at the $\ell+1$th iteration is shown as
\begin{align}\label{eq:procedure}
\mathrm{SINR}_1(\mat{w}_1^{(\ell)}, \mat{w}_2^{(\ell)})& \stackrel{\mat{w}_1^{(\ell+1)} = \Phi\left(\mat{w}_2^{(\ell)}\right)}{\longrightarrow} \mathrm{SINR}_1(\mat{w}_1^{(\ell+1)}, \mat{w}_2^{(\ell)}) ~~\nonumber \\
&\stackrel{\mat{w}_2^{(\ell+1)} = \Theta\left(\mat{w}_1^{(\ell+1)}\right)}{\longrightarrow} \mathrm{SINR}_1(\mat{w}_1^{(\ell+1)}, \mat{w}_2^{(\ell+1)}), \nonumber
\end{align}
In $\mat{w}_1^{(\ell+1)} = \Phi(\mat{w}_2^{(\ell)})$, since the feasible point set of $\mat{w}_1$ of $\mathrm{(P1)}$ includes $\mat{w}_1^{(\ell)}$ and additionally the global optimal solution to $\mathrm{(P1)}$ can be obtained by $\Phi(\cdot)$, it obviously implies $\mathrm{SINR}_1(\mat{w}_1^{(\ell+1)}, \mat{w}_2^{(\ell)}) = \mathrm{SINR}_1(\Phi\left(\mat{w}_2^{(\ell)}\right), \mat{w}_2^{(\ell)}) \geq \mathrm{SINR}_1(\mat{w}_1^{(\ell)}, \mat{w}_2^{(\ell)})$. Similarly, it is also easily verified $\mathrm{SINR}_1(\mat{w}_1^{(\ell+1)}, \mat{w}_2^{(\ell+1)})= \mathrm{SINR}_1(\mat{w}_1^{(\ell+1)}, \Theta\left(\mat{w}_1^{(\ell+1)}\right)) \geq \mathrm{SINR}_1(\mat{w}_1^{(\ell+1)}, \mat{w}_2^{(\ell)})$ by $\mat{w}_2^{(\ell+1)} = \Theta(\mat{w}_1^{(\ell+1)})$ due to the optimality of $\Theta(\cdot)$. As a consequence, the sequence of $\left\{\mathrm{SINR}_1(\mat{w}_1^{(\ell)}, \mat{w}_2^{(\ell)})\right\}_{\ell=1}^{\infty}$ \emph{monotonically} increases as the iteration number $\ell$ increases. In addition, since the sequence $\left\{\mathrm{SINR}_1(\mat{w}_1^{(\ell)}, \mat{w}_2^{(\ell)})\right\}_{\ell=1}^{\infty}$ is upper-bounded by the single-user SINR, i.e., $\frac{\lambda_1(\mat{H}_{11}^H\mat{H}_{11})}{\sigma_1^2}$ in (\ref{eq:SURate}), the convergence of the sequence $\left\{\mathrm{SINR}_1(\mat{w}_1^{(\ell)}, \mat{w}_2^{(\ell)})\right\}_{\ell=1}^{\infty}$, and thus the convergence of Algorithm \ref{alg:A1} is guaranteed for any feasible initial point $\mat{w}_2^{(0)}$.

Since the \emph{hard-coupled} two beamformers exist not only in the objective but also in the constraints in $\mathrm{(P0)}$, the conventional convergence analysis for the block coordinate descent algorithm \cite{ConvergenceStationaryP} that requires that the constraints are \emph{separable} among the variables is not applicable to our scenario. Therefore, it is unclear whether the proposed algorithm converges to a stationary point $\left(R_1(\mat{w}_1^{(\ell)}, \mat{w}_2^{(\ell)}), R_2(\mat{w}_1^{(\ell)}, \mat{w}_2^{(\ell)})\right)$ where $\mat{w}_1^{(\ell)}$ and $\mat{w}_2^{(\ell)}$ satisfy the KKT conditions of the original problem $\mathrm{(P0)}$.

\emph{\rmnum{2}) Quality of Solutions:} Due to the convergence of the proposed algorithm, the limit point of sequence of $\left\{\mathrm{SINR}_1(\mat{w}_1^{(\ell)}, \mat{w}_2^{(\ell)})\right\}_{\ell=0}^{\infty}$ for an arbitrary feasible initial $\mat{w}_2^{(0)}$ can be achieved by
\begin{align}
&\lim_{\ell \rightarrow \infty}\mathrm{SINR}_1(\mat{w}_1^{(\ell)}, \mat{w}_2^{(\ell)}) \nonumber \\ =&\mathrm{SINR}_1\Big(\Phi\left(\Theta\left(...\left(\Phi\left(\mat{w}_2^{(0)}\right)\right)\right)\right), \Theta\left(...\left(\mat{w}_2^{(0)}\right)\right)\Big). \nonumber
\end{align}
It implies that the performance of the alternating optimization algorithm depends on the initial beamformer $\mat{w}_2^{(0)}$. Denote the global optimal solution to $\mathrm{(P0)}$ by $(\mat{w}_1^\star, \mat{w}_2^\star)$. Take an extreme example, if $\mat{w}_2^{(0)}=\mat{w}_2^\star$, we can obtain the global optimum $\mat{w}_1^\star = \Phi\left( \mat{w}_2^{(0)}\right)$ directly due to the optimality of $\Phi(\cdot)$. Therefore, a good initial beamformer could lead to high performance. However generally, it is difficult to find a good initial point \emph{efficiently} for such a complex multi-variable optimization problem. In order to improve the performance, a common way in references is to implement the alternating optimization algorithm with multiple random initializations and then select the one with the best performance. In this paper, we desire to design a scheme to generate a good initialization \emph{efficiently}.

Inspired by the idea in \cite{EduardMISOCharcter, BalancingMIMO}, we heuristically propose a transmit beamformer design scheme by balancing the "egoistic" and "altruistic" strategies as
\begin{equation}\label{eq:BlancingAltEgo}
\mat{w}_i(\xi_{i,1}, \xi_{i,2}) = \overrightarrow{\xi_{i,1}\mat{w}_i^{Ego} + \xi_{i,2}\mat{w}_i^{Alt}}, ~~~~ i = 1,2,
\end{equation}
where $\xi_{i,1}$ and $\xi_{i,2}$ are complex-valued parameters satisfying $|\xi_{i,1}| + |\xi_{i,2}| = 1$. In fact, this tradeoff scheme is reasonable. For instance, it is necessary to be Pareto optimal for the two-user MISO IC \cite{EduardMISOCharcter}, and its a similar form still provides a good performance in sum-rate maximization for the multi-user single-stream MIMO IC \cite{BalancingMIMO}. The following simulation results show that this characterization cannot exactly achieve the whole strict Pareto boundary for the two-user MIMO IC but still has a promising performance.
In particular, the two ending point of strict Pareto boundary $E1$ and $E2$ can be achieved exactly by $(\mat{w}_1(1, 0), \mat{w}_2(0, 1))$ and $(\mat{w}_1(0, 1), \mat{w}_2(1, 0))$, respectively.

If $({\mat{w}}_1^{(0)}, {\mat{w}}_2^{(0)})$ corresponds to the bound of (\ref{eq:BlancingAltEgo}) or of random beamforming, then the proposed algorithm must improve (or at least keep) the bound of (\ref{eq:BlancingAltEgo}) or of random beamforming. However, it is not efficient to find those beamforming pairs achieving the bound achieved by (\ref{eq:BlancingAltEgo}) or by the random beamforming. Therefore, there is no guarantee to say that the proposed algorithm with \emph{only one initial beamformer} always achieves an \emph{outer} boundary than the bound of (\ref{eq:BlancingAltEgo}) or of random beamforming, but its performance will increase with the number of initial beamformers (i.e., multiple initializations). Since the $2N_T$-dimensional real space of each complex $\mat{w}_i$ can be \emph{approximately} reduced to 3-dimensional real space (i.e., $|\xi_{i,1}|$, $|\xi_{i,2}|$ and the difference of the phases of $\xi_{i,1}$ and $\xi_{i,2}$ in (\ref{eq:BlancingAltEgo}) without significant performance loss. Therefore, the proposed algorithm with the proposed initialization in (\ref{eq:BlancingAltEgo}) is \emph{more efficient or likely} to achieve a good performance compared with a random initialization.

To further enhance the efficiency of initialization by (\ref{eq:BlancingAltEgo}), a real constant parameter (i.e., the proportion of the "egoistic" strategy) is employed to reduce (\ref{eq:BlancingAltEgo}) to
\begin{equation}\label{eq:InitializationW}
\mat{w}_i = \overrightarrow{\zeta \cdot\mat{w}_i^{Ego} + (1-\zeta) \cdot \mat{w}_i^{Alt}},~~~~ i = 1,2,
\end{equation}
where $\zeta = \frac{R_2^\star-\underline{R}_2}{\overline{R}_2-\underline{R}_2}$ is a constant for a given $R_2^\star$. If $\mat{w}_2\notin\mathcal{W_F}$, we reset $\zeta \in \frac{R_2^\star-\underline{R}_2}{\overline{R}_2-\underline{R}_2} + [-\nu,\nu]$ with $0< \nu \leq \min\Big\{\frac{R_2^\star-\underline{R}_2}{\overline{R}_2-\underline{R}_2}, \frac{\overline{R}_2-R_2^\star}{\overline{R}_2-\underline{R}_2}\Big\}$ until $\mat{w}_2\in\mathcal{W_F}$. If $\mat{w}_2$ is still infeasible, we choose a randomly generated $\mat{w}_2\in\mathcal{W_F}$ directly.

Note that the characterization in (\ref{eq:InitializationW}) \emph{directly} corresponds to its own bound. Thus, our proposed algorithm with the initialization by (\ref{eq:InitializationW}) always outperforms the bound by (\ref{eq:InitializationW}), which can serve as a lower bound of the proposed algorithm.

\emph{\rmnum{3}) Implementation with Transmitter Cooperation:} For the purpose of implementation without an authority{\footnote{The proposed algorithm can be also implemented in a centralized way with the aid of an authority who does the optimization of both $\mat{w}_1$ and $\mat{w}_2$ based on the global CSI collected through feedback links.}}, we assume that each transmitter knows perfect global CSI in a block-fading environment through training and feedback and is willing to cooperate with each other transmitter for information exchange via backhaul links. For the optimization of $\mat{w}_1$ at $\mathrm{TX}_1$, $\mathrm{TX}_1$ solving $\mathrm{(P1)}$ based on the updated $\mat{w}_2$ from $\mathrm{TX}_2$. Similarly, $\mathrm{TX}_2$ optimizes $\mat{w}_2$ by solving $\mathrm{(P3)}$ based on the updated $\mat{w}_1$ from $\mathrm{TX}_1$. The algorithm can be terminated once $\mathrm{TX}_1$ finds that convergence criterion is satisfied. 

\emph{\rmnum{4}) Complexity Analysis:}
For the proposed algorithm shown in Algorithm \ref{alg:A1}, each iteration involves solving two convex semidefinite relaxation (SDR) problems (i.e., $\mathrm{(P2)}$ and $\mathrm{(P5)}$) and two implementations of the matrix rank-one decomposition of $\mat{W}_i$. In \cite{SDRLuo}, it is shown that the complexity of solving the SDR is \emph{polynomial} in the problem size (i.e., $N_T$) and the number of constraints (denoted by $m$), i,e., $\mathcal{O}(\max{(m, N_T)}^4 N_T^{1/2} \log(1/\epsilon))$ given a solution accuracy $\epsilon > 0$. In this paper, we have $N_T\geq 3$, and $m=2$ in $\mathrm{(P2)}$ and $m=3$ in $\mathrm{(P5)}$. Thus, the complexity of solving $\mathrm{(P2)}$ and $\mathrm{(P5)}$ is $\mathcal{O}(N_T^{4.5} \log(1/\epsilon))$. In terms of the complexity of the matrix rank-one decomposition \cite{SDPHuang2010}, the rank-one solution can be extracted in \emph{polynomial-time} if $\mathrm{Rank}(\mat{W}_i)\geq3$; If $\mathrm{Rank}(\mat{W}_i)\geq2$, it is sufficient to seek for a rank-one solution to a sequence of linear matrix equations within a slightly expanded range space of $\mat{W}_i$. If $\mathrm{Rank}(\mat{W}_i)=1$, only eigen-decomposition of $\mat{W}_i$ is needed.

In the following simulations, the average time of an iteration of Algorithm \ref{alg:A1} is 0.6180 seconds by running the MATLAB 7.10 on the computer with AMD Athlon(TM) 64 Processor 3200+, 2.01 GHZ and 2GB RAM. Additionally, the fast convergent behavior of the proposed algorithm is implied numerically (e.g., Fig. \ref{fig:no2-b} with 8.55 iterations on average and Fig. \ref{fig:no3-b} with 5.16 iterations on average). Therefore, the proposed algorithm has reasonable complexity.

\subsection{Extension to the multi-user MIMO IC}

Consider the $K$-user single-stream MIMO IC. With the MMSE receiver, the achievable rate of the link $\mathrm{TX}_k \mapsto \mathrm{RX}_k$ is expressed as ${R}_k(\{\mat{w}_k\}_{\mathcal{K}}) = \log_2\big(1+\mathrm{SINR}_k(\{\mat{w}_k\}_{\mathcal{K}})\big)~\forall k\in\mathcal{K}=\{1,...,K\}$,
where 
\begin{align}\label{eq:KUserMMSESINR}
&\mathrm{SINR}_k(\{\mat{w}_i\}_{\mathcal{K}})= \nonumber \\
&\mat{w}_k^H \underbrace{\mat{H}_{kk}^H(\sum_{i\neq k}\mat{H}_{ik}{\mat{w}_i}{\mat{w}_i}^H{\mat{H}_{ik}}^H+{\sigma}_k^2\mat{I})^{-1}\mat{H}_{kk}}_{\mat{A}_k(\mat{w}_{-k})}\mat{w}_k.
\end{align}
is the SINR expression of the $k$th user. $\mat{w}_{-k}$ denotes $\{\mat{w}_{i}\}_{\mathcal{K}\backslash\{k\}}$.

Without loss of generality, the optimization problem $\mathrm{(P0)}$ in the two-user case can be generalized to{\footnote{For multi-user case, an arbitrary Pareto-optimal point of the utility region can be achieved by maximizing one user's utility while fixing the others.}}
\begin{equation*}\label{eq:KP0}
\mathrm{(Q0)}
\left\{
\begin{aligned}
\max_{\{\mat{w}_i\}_{\mathcal{K}}} ~~&{\mathrm{SINR}_1(\{\mat{w}_i\}_{\mathcal{K}})} \nonumber \\
\mathrm{s.t.}~~~~&{\mathrm{SINR}_k(\{\mat{w}_i\}_{\mathcal{K}})} =  \mathrm{SINR}_k^\star,~~~\forall k\in{\mathcal{K}}\backslash\{1\}. \\
&\mat{w}_i^H\mat{w}_i\leq 1,~~~\forall i\in \mathcal{K},
\end{aligned}
\right.
\end{equation*}
where $\{\mathrm{SINR}_{i}^\star\}_{\mathcal{K}\backslash\{1\}}$ are assumed to be feasible. The problem $\mathrm{(Q0)}$ is a non-convex problem of $\{\mat{w}_i\}_{\mathcal{K}}$.

To extend the proposed algorithm to $\mathrm{(Q0)}$, we derive equivalent expressions of $\mathrm{SINR}_k(\{\mat{w}_i\}_{\mathcal{K}})$ in (\ref{eq:KUserMMSESINR}) by defining ${\mat{D(\mat{w}_{-i-k})}}=\sum_{\ell\neq k, i}\mat{H}_{\ell k}{\mat{w}_\ell}{\mat{w}_\ell}^H{\mat{H}_{\ell k}}^H+{\sigma}_k^2\mat{I}$:
\begin{align}\label{eq:MMSESINR2}
&\mathrm{SINR}_k(\{\mat{w}_i\}_{\mathcal{K}}) \nonumber \\
&=\mat{w}_k^H\mat{H}_{kk}^H\left(\mat{D(\mat{w}_{-i-k})} + \mat{H}_{ik}{\mat{w}_i}{\mat{w}_i}^H{\mat{H}_{ik}}^H \right)^{-1}\mat{H}_{kk}\mat{w}_k \nonumber\\
&\stackrel{(a)}{=} \mat{w}_k^H\mat{H}_{kk}^H{\mat{D}(\mat{w}_{-k-i})}^{-1}\mat{H}_{kk}\mat{w}_k - \nonumber\\
&\frac{\mat{w}_i^H\overbrace{\mat{H}_{ik}^H{\mat{D}(\mat{w}_{-k-i})}^{-1}\mat{H}_{kk}\mat{w}_k\mat{w}_k^H \mat{H}_{kk}^H{\mat{D}(\mat{w}_{-i-k})}^{-1}\mat{H}_{ik}}^{\mat{F}_k(\mat{w}_{-i})}\mat{w}_i}{1+{\mat{w}_i}^H\underbrace{{\mat{H}_{ik}}^H{\mat{D}(\mat{w}_{-i-k})}^{-1}\mat{H}_{ik}}_{\mat{G}_k(\mat{w}_{-i-k})}{\mat{w}_i}}
\end{align}
where the transformation $(a)$ is based on the Sherman-Morrison Formula \cite{MatrixAnalysisand} and $\mat{w}_{-i-k} = \{\mat{w}_\ell\}_{\mathcal{K}\backslash \{i,k\}}$.

Therefore, maximization of the objective function in $\mathrm{(Q0)}$ w.r.t. different beamformers is equivalent to
\begin{align}\label{eq:ObjEqui}
&\max_{\mat{w}_1}~\mathrm{SINR}_1(\{\mat{w}_i\}_{\mathcal{K}})
~{\Longleftrightarrow}~
\max_{\mat{w}_1} ~{\mat{w}_1}^H\mat{A}_1(\mat{w}_{-1}){\mat{w}_1} \nonumber \\
&\max_{\mat{w}_k}~\mathrm{SINR}_1(\{\mat{w}_i\}_{\mathcal{K}})
\stackrel{(b)}{\Longleftrightarrow}
\min_{\mat{w}_k} ~\frac{{\mat{w}_k}^H\mat{F}_k(\mat{w}_{-k}){\mat{w}_k}}{1+{\mat{w}_k}^H\mat{G}_k(\mat{w}_{-1-k}){\mat{w}_k}}, \nonumber \\
&~~~~~~~~~~~~~~~~~~~~~~~~~~~~~~~~~~~~~~~~\forall k\in{\mathcal{K}\backslash \{1\}},
\end{align}
and individual SINR constraint is equivalent to
\begin{align}\label{eq:ConEqui}
&\mathrm{SINR}_k(\{\mat{w}_i\}_{\mathcal{K}})=\mathrm{SINR}_k^\star \nonumber \\
{\Longleftrightarrow}~
&{\mat{w}_k}^H\mat{A}_k(\mat{w}_{-k}){\mat{w}_k} = \mathrm{SINR}_k^\star, \nonumber \\
\stackrel{(b)}{\Longleftrightarrow}~
&{\mat{w}_i}^H\mat{E}_k(\mat{w}_{-i}){\mat{w}_i} = \gamma_k(\mat{w}_{-i}),~~\forall i\neq k
\end{align}
where $\mat{E}_k(\mat{w}_{-i})$ and $\gamma_k(\mat{w}_{-i})$ are defined as
\begin{align}
&\mat{E}_k(\mat{w}_{-i})~\stackrel{\Delta}{=}~\mat{F}_k(\mat{w}_{-i}) - \nonumber \\
&\left(\mat{w}_k^H \mat{H}_{kk}^H{\mat{D}(\mat{w}_{-k-i})}^{-1}\mat{H}_{kk}\mat{w}_k - \mathrm{SINR}_k^\star\right) \mat{G}_k(\mat{w}_{-i-k}); \nonumber \\
&\gamma_k(\mat{w}_{-i}) ~\stackrel{\Delta}{=}~ \mat{w}_k^H \mat{H}_{kk}^H{\mat{D}(\mat{w}_{-k-i})}^{-1}\mat{H}_{kk}\mat{w}_k - \mathrm{SINR}_k^\star. \nonumber
\end{align}
The equivalence $(b)$ in both (\ref{eq:ObjEqui}) and (\ref{eq:ConEqui}) is based on (\ref{eq:MMSESINR2}).

\subsubsection{Optimization of $\mat{w}_1$}
Given the fixed $\mat{w}_{-1}$ and based on the equivalence results in (\ref{eq:ObjEqui}) and (\ref{eq:ConEqui}), $\mathrm{(Q0)}$ w.r.t. $\mat{w}_1$ is equivalent to
\begin{equation*}\label{eq:KP0}
\mathrm{(Q1)}
\left\{
\begin{aligned}
\max_{\mat{w}_1} ~~& {\mat{w}_1}^H\mat{A}_1(\mat{w}_{-1}){\mat{w}_1} \nonumber \\
\mathrm{s.t.}~~~&{\mat{w}_1}^H\mat{E}_k(\mat{w}_{-1}){\mat{w}_1} = \gamma_k(\mat{w}_{-1}),~\forall k\in{\mathcal{K}\backslash \{1\}}. \\
&{\mat{w}_1}^H{\mat{w}_1} \leq 1.
\end{aligned}
\right.
\end{equation*}
Observe that $\mathrm{(Q1)}$ is a homogeneous QCQP. By the SDR, $\mathrm{(Q1)}$ is relaxed to
\begin{equation*}\label{eq:KP0}
\mathrm{(Q1')}
\left\{
\begin{aligned}
\max_{\mat{W}_1\succeq \mat{0}} ~~& \mathrm{Tr}(\mat{A}_1(\mat{w}_{-1}){\mat{W}_1}) \nonumber \\
\mathrm{s.t.}~~~&\mathrm{Tr}(\mat{E}_k(\mat{w}_{-1}){\mat{W}_1}) = \gamma_k(\mat{w}_{-1}),~\forall k\in{\mathcal{K}\backslash \{1\}} \\
&\mathrm{Tr}({\mat{W}_1}) \leq 1
\end{aligned}
\right.
\end{equation*}
where $\mat{W}_1 = \mat{w}_1 \mat{w}_1^H$. Now, $\mathrm{(Q1')}$ becomes a convex problem w.r.t. $\mat{W}_1$. The optimal $\mat{W}_1^\star$ to $\mathrm{(Q1')}$ can be efficiently solved by a convex optimization toolbox.

If $\mathrm{Rank}(\mat{W}_1^\star) = 1$, the optimal rank-one solution is $\mat{w}_1 = \mat{u}_1(\mat{W}_1^\star)$. Otherwise,
we observe that $\mathrm{(Q1)}$ as a homogeneous QCQP has $K$ constraints, and thus an optimal $\mat{w}_1$ to $\mathrm{(Q1)}$ can be reconstructed from $\mat{W}_1^\star$ for $K\leq 4$ by the matrix rank-one decomposition method in \cite{SDPHuang2011}. When $K\geq 5$, there exist several approaches (e.g., the eigenvector approximation method and the randomization method) to extract an \emph{approximate} $\mat{w}_1$ from $\mat{W}_1^\star$. Although these approximation methods are not tight, intensive research show that they provide promising performance (the interested readers could refer to the analysis of approximation accuracy bounds in \cite{SDRLuo}).

\subsubsection{Optimization of $\mat{w}_k,~\forall k\neq 1$}
Given the fixed $\mat{w}_{-k}$ and based on the equivalence results in (\ref{eq:ObjEqui}) and (\ref{eq:ConEqui}),
$\mathrm{(Q0)}$ is equivalent to
\begin{equation*}\label{eq:KP0}
\mathrm{(Qk)}
\left\{
\begin{aligned}
\min_{\mat{w}_k} ~~&\frac{{\mat{w}_k}^H\mat{F}_k(\mat{w}_{-k}){\mat{w}_k}}{1+{\mat{w}_k}^H\mat{G}_k(\mat{w}_{-1-k}){\mat{w}_k}} \nonumber \\
\mathrm{s.t.}~~~&{\mat{w}_k}^H\mat{A}_k(\mat{w}_{-k}){\mat{w}_k} = \mathrm{SINR}_k^\star \nonumber \\
&{\mat{w}_k}^H\mat{E}_\ell(\mat{w}_{-k}){\mat{w}_k} = \gamma_\ell(\mat{w}_{-k}),~\forall \ell\in{\mathcal{K}\backslash \{1,k\}} \\
&{\mat{w}_k}^H{\mat{w}_k} \leq 1.
\end{aligned}
\right.
\end{equation*}
Observe that the objective function belongs to fractional program, while it is not a quasi-convex function due to the convexity of both the nominator function and the denominator function.
To deal with this problem, we transform the problem $\mathrm{(Qk)}$ via the SDR to
\begin{equation*}\label{eq:KP0}
\mathrm{(Qk')}
\left\{
\begin{aligned}
\min_{\mat{W}_k\succeq\mat{0}} &\frac{\mathrm{Tr}(\mat{F}_k(\mat{w}_{-k}){\mat{W}_k})}{1+\mathrm{Tr}(\mat{G}_k(\mat{w}_{-1-k}){\mat{W}_k})} \nonumber \\
\mathrm{s.t.}~&\mathrm{Tr}(\mat{A}_k(\mat{w}_{-k}){\mat{W}_k}) = \mathrm{SINR}_k^\star \nonumber \\
&\mathrm{Tr}(\mat{E}_\ell(\mat{w}_{-k}){\mat{W}_k}) = \gamma_\ell(\mat{w}_{-k}),\forall \ell\in{\mathcal{K}\backslash \{1,k\}} \\
&\mathrm{Tr}({\mat{W}_k}) \leq 1.
\end{aligned}
\right.
\end{equation*}
where $\mat{W}_k = \mat{w}_k\mat{w}_k^H$. It is known that full power transmission is not always Pareto-optimal for the general multi-user MIMO/MISO IC (related to the number of users and transmit/receive antennas), which is different from the two-user Pareto-optimal full power transmission (Proposition 2). It leads to $\frac{\mathrm{Tr}(\mat{F}_k(\mat{w}_{-k}){\mat{W}_k})}{1+\mathrm{Tr}(\mat{G}_k(\mat{w}_{-1-k}){\mat{W}_k})} \neq \frac{\mathrm{Tr}(\mat{F}_k(\mat{w}_{-k}){\mat{W}_k})}{\mathrm{Tr}(\left(\mat{I}+\mat{G}_k(\mat{w}_{-1-k})\right){\mat{W}_k})}$. Thus, the Charnes-Cooper variable transformation used in the optimization problem $\mathrm{(P4)}$ in the two-user case is not applicable to $\mathrm{(Qk')}$ any longer. Nevertheless, we observe that both the nominator function and the denominator function of the objective function are non-negative, differentiable and affine with $\mat{W}_k$. By introducing a real scalar parameter $\mu_k \geq 0$, the fractional programming problem $\mathrm{(Qk')}$ is equivalent to a parametric programming problem
\begin{align}\label{eq:PGOF2}
\mathcal{F}(\mu_k)=&\min_{\mat{W}_k\in \mathcal{S}_{\mat{W}_k}}\Big\{{\mathrm{Tr}(\mat{F}_k(\mat{w}_{-k}){\mat{W}_k})} \nonumber \\
&~~~~~~~~-\mu_k \left({1+\mathrm{Tr}(\mat{G}_k(\mat{w}_{-1-k}){\mat{W}_k})}\right)\Big\},
\end{align}
where $\mathcal{S}_{\mat{W}_k}$ denotes the constraint set of $\mat{W}_k$ consisting of all the constraints in $\mathrm{(Qk')}$, and it is obvious that $\mathcal{S}_{\mat{W}_k}$ is a convex set. Assume the optimal solution to $\mathrm{(Qk')}$ is $\mat{W}_k^\star$.
If $\mu_k^\star = \frac{\mathrm{Tr}(\mat{F}_k(\mat{w}_{-k}){\mat{W}_k^\star})}{{1+\mathrm{Tr}(\mat{G}_k(\mat{w}_{-1-k}){\mat{W}_k^\star})}}$, it implies $\mathcal{F}(\mu_k^\star)=0$. Thus, solving $\mathrm{(Qk')}$ is equivalent to finding the root of the equation $\mathcal{F}(\mu_k)=0$.

Given $\mu_k$, (\ref{eq:PGOF2}) is a convex optimization problem w.r.t. $\mat{W}_k$, and its optimal solution $\mat{W}_k^\star(\mu_k)$ can be efficiently solved.
Therefore, $\mathcal{F}(\mu_k)=0$ can be further formulated as
\begin{align}\label{eq:PGOF4}
\mathcal{F}(\mu_k^\star)= &{\mathrm{Tr}(\mat{F}_k(\mat{w}_{-k}){\mat{W}_k^\star})} \nonumber \\ &~~~-\mu_k^\star\cdot\left({1+\mathrm{Tr}(\mat{G}_k(\mat{w}_{-1-k}){\mat{W}_k^\star})}\right)=0,
\end{align}
From \cite{NonlinearFracProg}, we know that $\mathcal{F}(\mu_k)$ is continuous, concave, strictly decreasing in $\mu_k$ and $\mathcal{F}(\mu_k)=0$ has a unique solution. Additionally, we find that $-\left({1+\mathrm{Tr}(\mat{G}_k(\mat{w}_{-1-k}){\mat{W}_k^\star})}\right)$ is a subgradient of $\mathcal{F}(\mu_k)$ for any $\mu_k$.
Thus, (\ref{eq:PGOF4}) can be solved by a generalized Newton method (also known as the Dinkelbach algorithm) described in Algorithm \ref{alg:A2}.
\begin{algorithm}
\DontPrintSemicolon
\KwIn{$\mu_k^{(0)}$ satisfying $\mathcal{F}(\mu_k^{(0)}) \leq 0$, tolerance $\epsilon$.}
\KwOut{Optimal $\mu_k^\star$ and $\mat{W}_k^\star$.}
\Begin{
$\ell=0$\;

\While{$|\mathcal{F}(\mu_k^{(\ell)})|> \epsilon$}{ \;

Given $\mu_k^{(\ell)}$, solve optimal $\mat{W}_k^\star(\mu_k^{(\ell)})$ to (\ref{eq:PGOF2});\;

$\mu_k^{(\ell+1)} = \frac{{\mathrm{Tr}(\mat{F}_k(\mat{w}_{-k}){\mat{W}_k^\star(\mu_k^{(\ell)})})}}{1+\mathrm{Tr}(\mat{G}_k(\mat{w}_{-1-k}){\mat{W}_k^\star(\mu_k^{(\ell)})})}$ \footnotemark;\;

$\ell++$.\;

}
$\mu_k^\star = \mu_k^{(\ell)}$ and $\mat{W}_k^\star = \mat{W}_k^\star(\mu_k^\star)$.\;
}
\caption{The generalized Newton method to solve (\ref{eq:PGOF2})}\label{alg:A2}
\end{algorithm}
\footnotetext{This generalized Newton iterative update is from $\mu_k^{(\ell+1)} \stackrel{\Delta}{=} \mu_k^{(\ell)} - \frac{\mathcal{F}(\mu_k^{(\ell)})}{-\left({1+\mathrm{Tr}(\mat{G}_k(\mat{w}_{-1-k}){\mat{W}_k^\star(\mu_k^{(\ell)})})}\right)} = \mu_k^{(\ell)} - \frac{{\mathrm{Tr}(\mat{F}_k(\mat{w}_{-k}){\mat{W}_k^\star(\mu_k^{(\ell)})})} -\mu_k^{(\ell)}\left({1+\mathrm{Tr}(\mat{G}_k(\mat{w}_{-1-k}){\mat{W}_k^\star(\mu_k^{(\ell)})})}\right)}{-\left({1+\mathrm{Tr}(\mat{G}_k(\mat{w}_{-1-k}){\mat{W}_k^\star(\mu_k^{(\ell)})})}\right)} =\frac{{\mathrm{Tr}(\mat{F}_k(\mat{w}_{-k}){\mat{W}_k^\star(\mu_k^{(\ell)})})}}{1+\mathrm{Tr}(\mat{G}_k(\mat{w}_{-1-k}){\mat{W}_k^\star(\mu_k^{(\ell)})})}$.}

The algorithm as a Newton procedure to determine the root of the equation $\mathcal{F}(\mu_k^\star)=0$ has superlinear convergence. By the Algorithm \ref{alg:A2}, the optimal $\mu_k^\star$ and $\mat{W}_k^\star(\mu_k^\star)$ to (\ref{eq:PGOF2}) is obtained. Equivalently, $\mat{W}_k^\star(\mu_k^\star)$ is an optimal solution to $\mathrm{(Qk')}$ \cite{NonlinearFracProg}. Then, a tight (for $K\leq 4$) or an approximate (for $K\geq5$) solution $\mat{w}_k$ to $\mathrm{(Qk)}$ can be extracted from $\mat{W}_k^\star$.

Above all, the proposed alternating optimization algorithm extended to solve $\mathrm{(Q0)}$ can be described
as Algorithm 3.
\begin{algorithm}
\DontPrintSemicolon
\KwIn{$\{R_k^\star\}_{{\mathcal{K}\backslash \{1\}}}$ where $R_k^\star=\log_2\left(1+\mathrm{SINR}_k^\star\right)$ and $\mathrm{SINR}_k^\star \in \Big(0, \frac{1}{\sigma_k^2}\lambda_1(\mat{H}_{kk}\mat{H}_{kk}^H)\Big]$.}
\KwOut{ A convergent point $(R_1^{(\ell)}, R_2^\star,...,R_K^\star)$ with $\{\mat{w}_i^{(\ell)}\}_{\mathcal{K}}$.}
\Begin{
Initialization:\;
Set a feasible $\mat{w}_{-1}^{(0)}$, $\ell=0$.

\While{some convergence criterion is not satisfied}{ $\ell++$.\;

\For{$k = 1 \to K$}{

Given $\mat{w}_{-k}^{(\ell-1)}$, obtain an optimal $\mat{W}_k$ to $\mathrm{(Qk')}$; \;

Extract a tight/approximate $\mat{w}_k$ from $\mat{W}_k$ to $\mathrm{(Qk)}$. \;

\If{$K\geq 5~\mathrm{and}~\mathrm{SINR}_1\left(\mat{w}_{k}^{(\ell)}, \mat{w}_{-k}^{(\ell-1)/(\ell)}\right) < \mathrm{SINR}_1\left(\mat{w}_{k}^{(\ell-1)}, \mat{w}_{-k}^{(\ell-1)/(\ell)}\right)$\footnotemark}
{
$\mat{w}_k^{(\ell)} = \mat{w}_k^{(\ell-1)}$; \;
}
}

Compute $R_1^{(\ell)} = \log_2\left(1 + \mathrm{SINR}_1\left(\{\mat{w}_k^{(\ell)}\}_{\mathcal{K}}\right)\right)$. \;
}
}
\caption{$K$-User Alternating Optimization Algorithm}\label{alg:A3}
\end{algorithm}
\footnotetext{This condition is to make sure that a better (at least the same) solution to $\mathrm{(Qk)}$ (only for $K\geq 5$) is always obtained in each iteration such that the objective function's non-decreasing convergence is guaranteed.}

\begin{remark}
The proposed alternating optimization algorithm can be extended to the $K$-user MIMO IC. For $K\leq 4$, it is the same as the two-user case that each optimal single-beamformer can be obtained in each iteration. For $K\geq5$, each approximate optimal single-beamformer is obtained in each iteration. Following the same line of the proof of the Algorithm \ref{alg:A1}'s convergence in Section \Rmnum{4}-C-2, the convergence of the Algorithm \ref{alg:A3} is also guaranteed.
\end{remark}

\section{Illustrations and Discussions}\label{sec:Illu}
To illustrate the achievable rate region by the proposed algorithm, we consider a two-user Gaussian MIMO IC, where $N_T=3$ and $N_R=2$. The transmit power budget is set to 1 for the two users, and noise power $\sigma_1^2 = \sigma_2^2 = 10^{-\frac{\mathrm{SNR}}{10}}$ where SNR=10dB. The channels $\mat{H}_{11}, \mat{H}_{12}, \mat{H}_{21}$ and $\mat{H}_{22}$ are

\[
\left.
\begin{array}{c}
\tiny{
\begin{pmatrix}
-0.3034 + 1.9096i&-0.3790 + 0.4201i&0.0357 + 0.7337i \\
-0.6358 - 0.8030i&-0.7881 - 0.1273i&0.7534 + 0.8348i
\end{pmatrix}
}
\end{array},
\right.
\]
\[
\left.
\begin{array}{c}
\tiny{
\begin{pmatrix}
-0.6758 + 0.1040i  &-0.5949 - 0.0344i   &0.4311 + 0.9658i \\
-2.1621 + 0.5404i  &-0.0037 + 0.6627i   &0.8611 + 1.2318i
\end{pmatrix}
}
\end{array},
\right.
\]
\[
\left.
\begin{array}{c}
\tiny{
\begin{pmatrix}
0.3999 + 0.1567i &0.3798 - 0.5619i  &-0.1005 + 0.2836i \\
-0.5494 - 0.4648i&1.1971 - 0.5297i  &-0.7271 + 0.2114i
\end{pmatrix}
}
\end{array},
\right.
\]
\[
\left.
\begin{array}{c}
\tiny{
\begin{pmatrix}
  -0.0308 - 0.1133i   &0.0433 - 0.3313i   &0.3047 - 1.2157i \\
  -1.4947 - 1.8676i   &-0.9430 + 0.5704i  &-1.3328 + 1.4638i
\end{pmatrix}
}
\end{array}.
\right.
\]

\subsection{Convergence and Performance of Initialization by Eq.(\ref{eq:InitializationW})}

To study the convergence rate of the alternating optimization
algorithm and evaluate the effectiveness of the initialization in (\ref{eq:InitializationW}), we respectively use (\ref{eq:InitializationW}) and 200 randomly generated feasible normalized vectors as initial $\mat{w}_2$. Then, we run  Algorithm \ref{alg:A1} until $|R_1^{(\ell)}-R_1^{(\ell-1)}|\leq 10^{-3}$.
\begin{figure}%
\centering
\subfigure[][]{%
\includegraphics[scale=0.45]{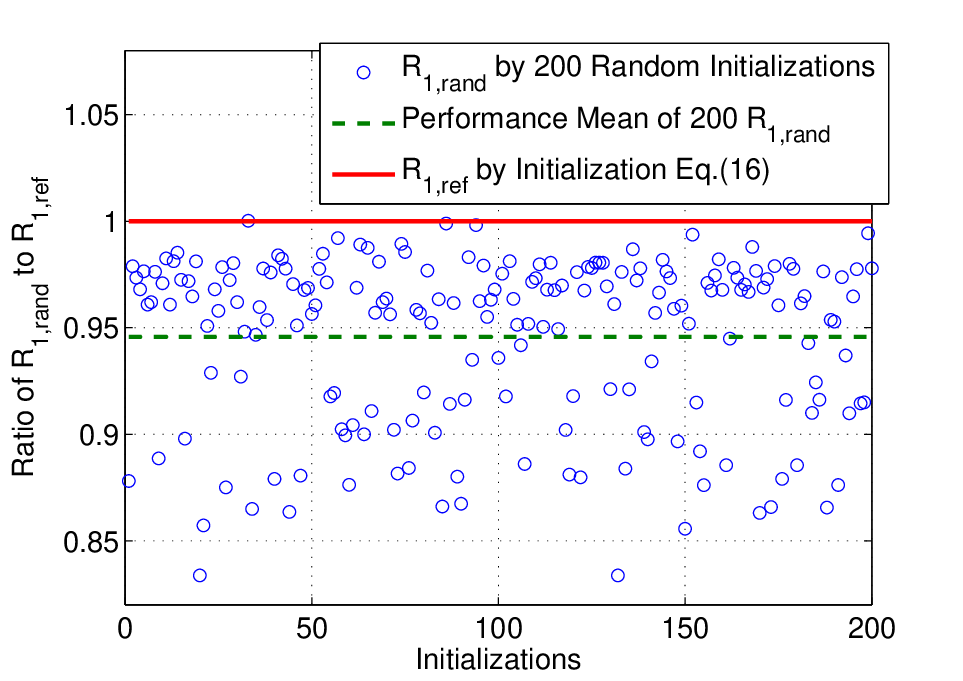}\label{fig:no2-a}
}%
\hspace{11pt}%
\subfigure[][]{%
\includegraphics[scale=0.45]{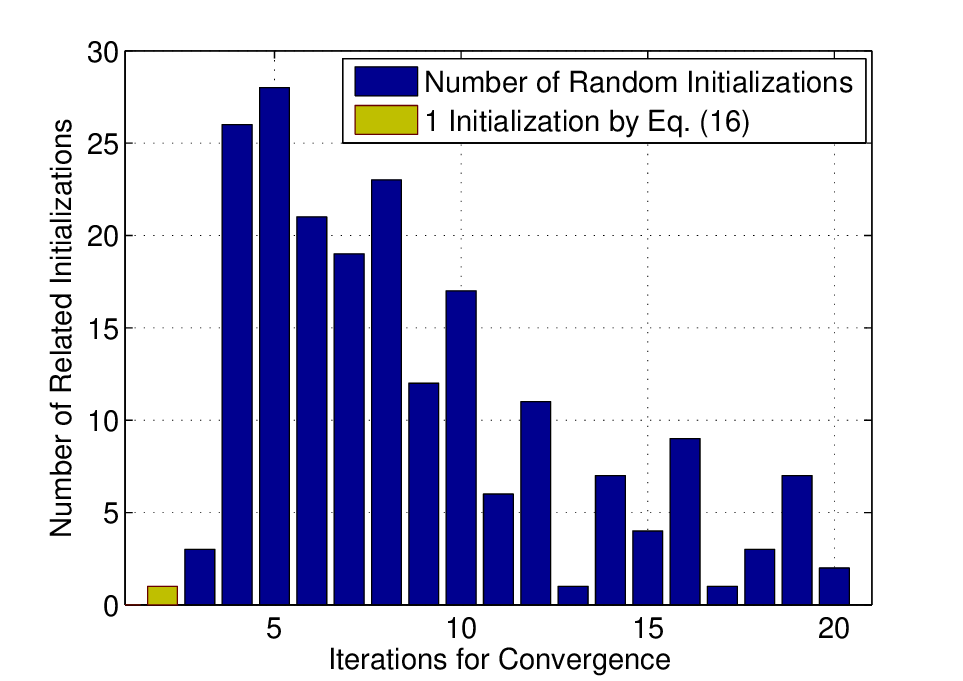}\label{fig:no2-b}
}
\caption[A set of two subfigures.]{Performance of Algorithm \ref{alg:A1} by different initializations at $R_2^\star = 5.6398$:
\subref{fig:no2-a} Performance comparison of 200 random initializations and 1 initialization by Eq. (\ref{eq:InitializationW});
\subref{fig:no2-b} Convergence behavior of Algorithm \ref{alg:A1}.}\label{fig:no2}
\end{figure}

\begin{figure}%
\centering
\subfigure[][]{%
\includegraphics[scale=0.45]{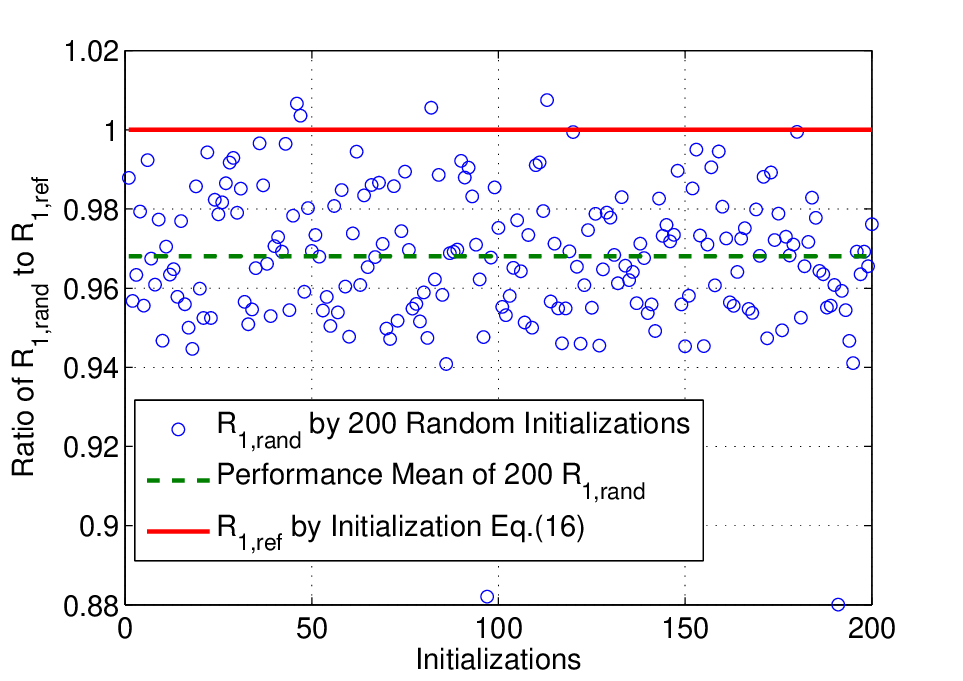}\label{fig:no3-a}
}%
\hspace{11pt}%
\subfigure[][]{%
\includegraphics[scale=0.45]{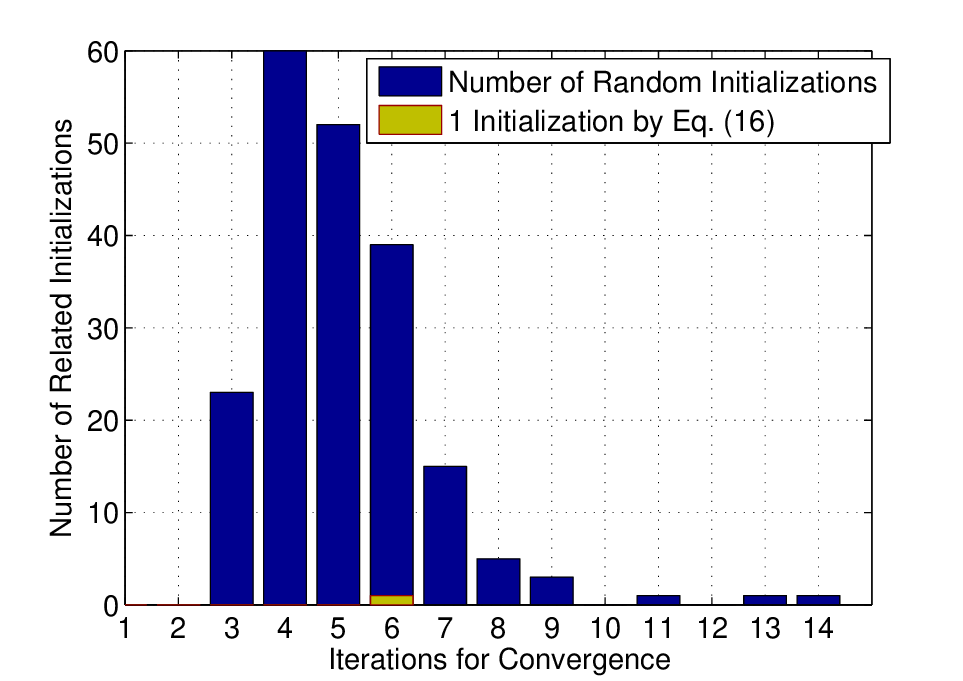}\label{fig:no3-b}
}
\caption[A set of two subfigures.]{Performance of Algorithm \ref{alg:A1} by different initializations at $R_2^\star = 6.2898$:
\subref{fig:no3-a} Performance comparison of 200 random initializations and 1 initialization by Eq. (\ref{eq:InitializationW}) ;
\subref{fig:no3-b} Convergence behavior of Algorithm \ref{alg:A1}.}\label{fig:no3}
\end{figure}

Fig. \ref{fig:no2} and Fig. \ref{fig:no3} show the performance of the alternating optimization algorithm by different initializations. More precisely, Fig. \ref{fig:no2} is for a given $R_2^\star=\underline{R}_2 + \frac{2}{19}\cdot(\overline{R}_2-\underline{R}_2) = 5.6398$ (close to the ending point $(\overline{R}_1, \underline{R}_2)$). Fig. \ref{fig:no2-a} shows that the achieved $R_1$ with initialization by (\ref{eq:InitializationW}) nearly always outperforms that with 200 random initializations. Fig. \ref{fig:no2-b} implies the convergence rate of the algorithm. Fig. \ref{fig:no3} is for a given $R_2^\star=\underline{R}_2 + \frac{11}{19}\cdot\left(\overline{R}_2-\underline{R}_2\right) = 6.2898$ (corresponding to the middle of strict Pareto boundary). Similarly, Fig. \ref{fig:no3-a} and Fig. \ref{fig:no3-b} also show that initialization by (\ref{eq:InitializationW}) has a promising performance and fast convergence behavior. Therefore, simulation results imply that (\ref{eq:InitializationW}) is a good choice for initialization.

\subsection{Performance Comparison}

Fig. \ref{fig:no4} illustrates the achievable boundary by the proposed algorithm, i.e., Algorithm \ref{alg:A1}, compared with the existing methods. {The term "Proposed\_{1+9}" denotes the best result obtained by running Algorithm \ref{alg:A1} with 1 initialization in (\ref{eq:InitializationW}) and 9 random initializations, while "Proposed\_1" represents the result only with 1 initialization in (\ref{eq:InitializationW}). The SINR targets are $\mathrm{SINR}_2^\star=2^{R_2^\star-1}$s where $R_2^\star=\underline{R}_2 + \frac{n}{50}\cdot{(\overline{R}_2 - \underline{R}_2)},~ n=1,2,...,49$.} Similarly, the term "WMMSE\_10" denotes the weighted sum rate maximization algorithm \cite{MIMOIBCLuo} with 10 random initializations, where weighted sum rate is expressed as $w\cdot R_1 + (1-{w})\cdot R_2$. with the weights $w$s in $[0.05:0.05:0.95]$. "WMMSE\_1" denotes the result only with 1 initialization by (\ref{eq:InitializationW}), i.e, setting the initial beamformers as $\mat{w}_1^{(0)} = \overrightarrow{w\cdot \mat{w}_1^{Ego}+(1-w)\cdot \mat{w}_1^{Alt}}$ and $\mat{w}_2^{(0)} = \overrightarrow{w\cdot \mat{w}_2^{Ego}+(1-w)\cdot \mat{w}_2^{Alt}}$. The curve denoted by "RandBeam\_10mil" means the outermost boundary of the rate region achieved by 10 million random normalized transmit beamformer pairs (each receiver is the MMSE filter). Theoretically, if \emph{exhaustive} random beamforming pairs are chosen, its bound is exactly the Pareto boundary. However, there exist infinite random beamforming pairs so that we choose as many as 100 million random beamforming pairs in the simulations to serve as an \emph{approximate Pareto boundary}. The term "SimpleReceiver" is the outermost boundary of the region achieved by \cite{EfficientComputationMISO} where each receiver is fixed as the largest left singular vector of the corresponding direct channel matrix. The curve denoted by "Eq.(\ref{eq:BlancingAltEgo})" is the outermost boundary of the achieved region by (\ref{eq:BlancingAltEgo}) with complex-valued parameters by 3-dimensional grid search. The boundary of "Eq.(\ref{eq:InitializationW})" illustrates (\ref{eq:InitializationW}) with $\zeta = \frac{n-1}{N},~n=1,...,N+1$ where $N=100$. "ZF points" denotes two outmost points of the ZF points by Eq. (\ref{eq:RateZF}).

If we consider the curve by "Proposed\_{1+9}" as a nearly optimal boundary, we find that "Proposed\_1" has a promising/robust performance only with one initialization by (\ref{eq:InitializationW}). Also observe that the proposed algorithm "Proposed\_{1+9}" and "Proposed\_1", and "WMMSE\_10" yield a similar performance at convex parts of boundary and outperform the others under the same accuracy for convergence.
However, since the weighted sum maximization method cannot achieve the non-convex boundary, and even the achieved points on the convex boundary are still unevenly distributed. This is why the "WMMSE\_10" as a weighted sum maximization method does not achieve the part between "P1" and "P2" in Fig. \ref{fig:no4}.

To further evaluate the performance of "WMMSE\_10" and "Proposed\_{1+9}" on illustrating the Pareto boundary, another simulation is done and shown in Fig. \ref{fig:no5}. Even with fine weights ${w}$s in $[0.05:0.005:0.95]$, we find that there exists a large jump between the points "P3" and "P4" by "WMMSE\_10" so that the rate region cannot be illustrated effectively.
\begin{figure}
\begin{center}
\includegraphics[scale=0.48]{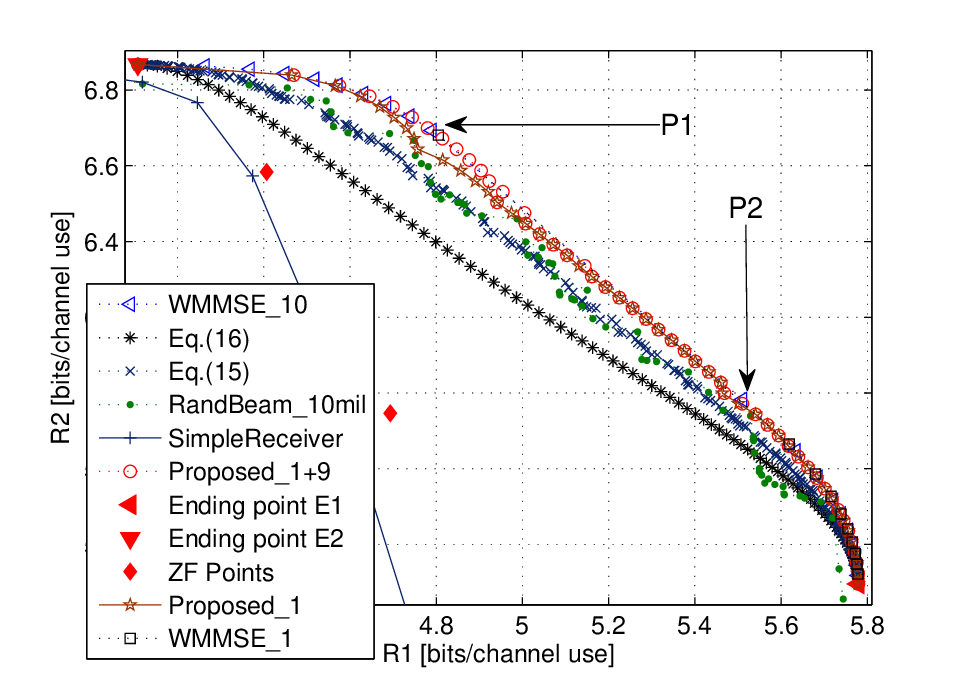}
\caption{Achievable boundary with SNR=10dB and $N_T=3, N_R=2$ for a random Gaussian channel data}\label{fig:no4}
\end{center}
\end{figure}
\begin{figure}
\begin{center}
\includegraphics[scale=0.48]{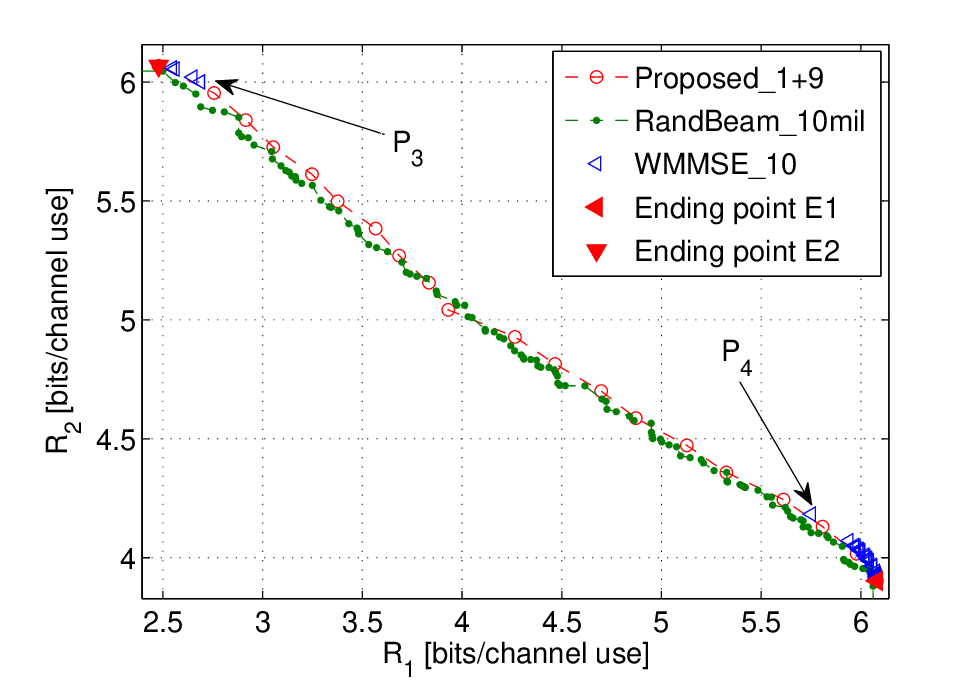}
\caption{Achievable boundary with SNR=10dB and $N_T=3, N_R=2$ for another random Gaussian channel data}\label{fig:no5}
\end{center}
\end{figure}

In fact, only through the results by "WMMSE\_10", we do not know what the part between "P1" and "P2" in Fig. \ref{fig:no4} and the part between "P1" and "P2" in Fig. \ref{fig:no5} look like. If they are concave parts, the degree of concavity is still unknown. For our proposed algorithm, although its curve denoted by "Proposed\_{1+9}" is not guaranteed to exactly be the strict Pareto boundary, it has a even better performance than the \emph{approximate Pareto boundary} (i.e., "Rand\_10mil" curves in Fig. \ref{fig:no4} and Fig. \ref{fig:no5}). Thus, it is able to serve as a more reasonable/complete inner bound of the whole strict Pareto boundary, especially the non-convex part. This is a main advantage of the proposed algorithm to the existing algorithms.

\subsection{Illustration of the multi-user case}
In order to evaluate the performance of Algorithm \ref{alg:A3} for the multi-user MIMO IC, a three-user MIMO IC example is simulated.
Fig. \ref{fig:converg} shows the fast convergence behavior when $R_1$ is maximized with $(R_3, R_2) = (0.2700, 2.3720)$ and the convergence threshold $10^{-4}$.
In Fig. \ref{fig:3Dwith}, a three dimensional rate region is illustrated after computing $R_1$ for 65 samples of $(R_3, R_2)$ and interpolation.
\begin{figure}
\begin{center}
\includegraphics[scale=0.53]{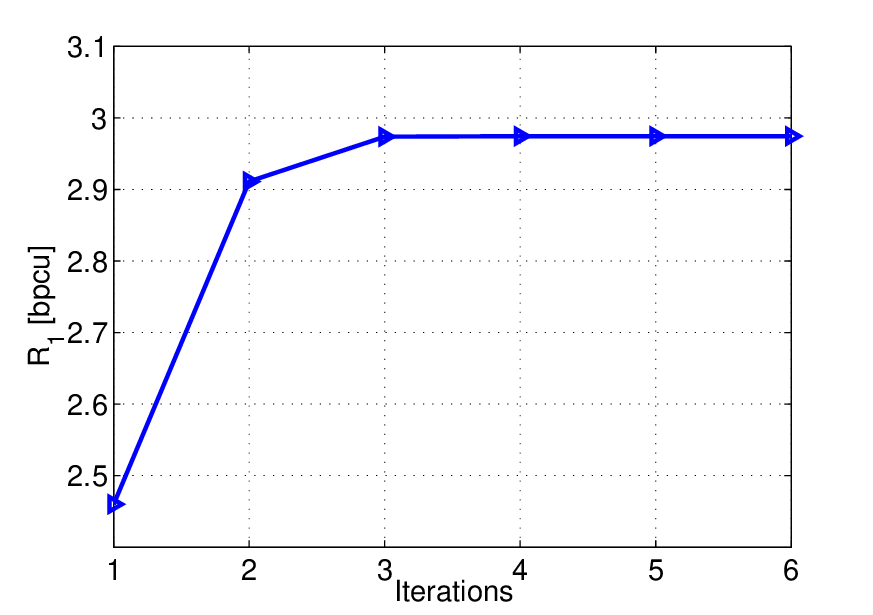}
\caption{Convergence behavior of Algorithm \ref{alg:A3} for three-user MIMO IC with SNR=0dB and $N_T=3, N_R=2$.}\label{fig:converg}
\end{center}
\end{figure}

\begin{figure}
\begin{center}
\includegraphics[scale=0.48]{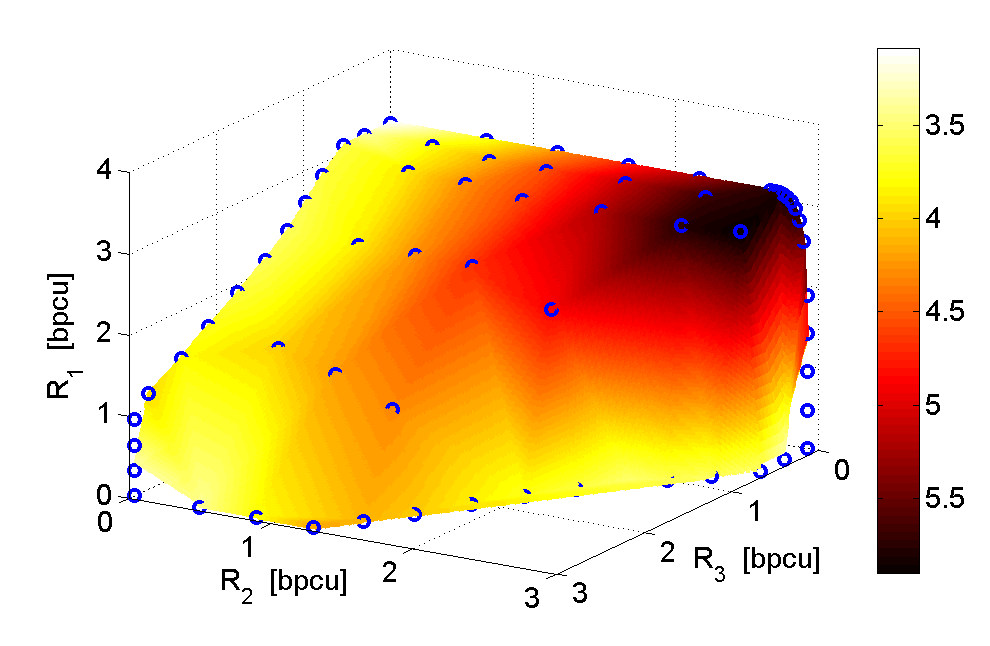}
\caption{Achievable rate region for three-user MIMO IC with SNR=0dB and $N_T=3, N_R=2$. The color bar shows the sum rate.}\label{fig:3Dwith}
\end{center}
\end{figure}

\section{Conclusions and Future Work}\label{sec:Conclu}
In this paper, the non-strict Pareto boundary and two ending points of strict Pareto boundary of the two-user scenario are computed exactly. To find the strict Pareto boundary for the two/multi-user single-beam MIMO IC, we formulate an optimization problem to maximize the rate of one user while the rates of the other users are fixed such that a bound point is reached. This problem is different from the traditional problems (e.g., maximize (weighted) sum-rate and max-min rate) and there exists little work on this type of problem in the MIMO IC. We propose an alternating optimization algorithm to solve this non-convex problem. The convergence of the proposed algorithm is guaranteed and a high quality suboptimal solution is obtained. Furthermore, the proposed computation algorithm has two main advantages: 1) This algorithm can be applied to those optimization problems with rate constraints (e.g. rate requirements of primary links in overlay cognitive radio environments or for private messages); 2) For the strict Pareto boundary, this algorithm is able to compute a high quality suboptimal operating point satisfying arbitrary rate requirements. A series of well-distributed rate requirements can lead to a sequence of \emph{well-distributed} operating points, which can serve as a reasonable and complete inner bound of the strict Pareto boundary.

The proposed algorithm requires that each transmitter knows perfect global CSI, which is challenging
for distributed systems. However, this work provides a benchmark for the algorithms that only imperfect/partial CSI is available at transmitters. Future work should focus on developing robust cooperative algorithms for the multi-user MIMO IC.

\appendices
\section{Proof of Proposition 1}\label{sec:proofProp1}
\begin{IEEEproof}
Given beamformers $\mat{w_1}$ and $\mat{w_2}$, according to the matrix inversion lemma \cite{MatrixAnalysisand}, (\ref{eq:MMSESINR}) can be rewritten as:
\begin{align}\label{eq:MMSE SINR2}
&~~\mathrm{SINR}_i(\mat{w}_1,\mat{w}_{2}) \nonumber \\
&= \left({\mat{H}_{ii}\mat{w}_i}\right)^H \left(\frac{1}{{\sigma}_i^2}\mat{I} - \frac{\mat{H}_{ki}{\mat{w}_k}(\mat{H}_{ki}{\mat{w}_k})^H}{{\sigma}_i^2({\sigma}_i^2+\lVert \mat{H}_{ki}{\mat{w}_k} \lVert^2)}\right){\mat{H}_{ii}\mat{w}_i} \nonumber \\
&= \frac{\lVert \mat{H}_{ii}{\mat{w}_i} \lVert^2}{\sigma_i^2} - \frac{\left|(\mat{H}_{ii}{\mat{w}_i})^H\mat{H}_{ki}{\mat{w}_k}\right|^2}{{{\sigma}_i^2({\sigma}_i^2+\lVert \mat{H}_{ki}{\mat{w}_k} \lVert^2)}} \nonumber \\
&= \frac{\lVert \mat{H}_{ii}{\mat{w}_i} \lVert^2}{\sigma_i^2} \cdot \left( 1 - \frac{\left|\overrightarrow{\mat{H}_{ii}{\mat{w}_i}}^H\cdot \overrightarrow{\mat{H}_{ki}{\mat{w}_k}}\right|^2\cdot\lVert \mat{H}_{ki}{\mat{w}_k} \lVert^2}{{{\sigma}_i^2+\lVert \mat{H}_{ki}{\mat{w}_k} \lVert^2}} \right) \nonumber \\
&= \left(1-\left|\overrightarrow{\mat{H}_{ii}{\mat{w}_i}}^H\cdot \overrightarrow{\mat{H}_{ki}{\mat{w}_k}}\right|^2\right)\cdot\frac{{\lVert \mat{H}_{ii}\mat{w}_i \lVert}^2}{\sigma_i^2} \nonumber \\
&~~~~~~~~+ \left|\overrightarrow{\mat{H}_{ii}{\mat{w}_i}}^H\cdot \overrightarrow{\mat{H}_{ki}{\mat{w}_k}}\right|^2\cdot \frac{{\lVert \mat{H}_{ii}\mat{w}_i \lVert}^2}{\sigma_i^2 + {\lVert \mat{H}_{ki}\mat{w}_k \lVert}^2}.
\end{align}

For two complex vectors $\mat{a}$ and $\mat{b}$, the cosine of the complex-valued angle between $\mat{a}$ and $\mat{b}$ is defined as \cite{CA}
$\cos(\theta_C)=\frac{\mat{a}^H\mat{b}}{\lVert \mat{a} \lVert \cdot \lVert \mat{b} \lVert}$
where $\cos(\theta_C)=\mu e^{j\psi}$ with $\mu = |\cos(\theta_C)|\leq1$ and $-\pi\leq\theta_C\leq\pi$ is called pseudo angle between $\mat{a}$ and $\mat{b}$.

The Hermitian angle between $\mat{a}$ and $\mat{b}$ is defined as
\begin{align}\label{eq:HermitianA4}
\cos(\theta_H)=|\cos(\theta_C)|=\frac{|\mat{a}^H\mat{b}|}{\lVert \mat{a} \lVert \cdot \lVert \mat{b} \lVert},~~~0\leq\theta_H\leq\pi/2.\nonumber
\end{align}
It implies $\big|\overrightarrow{\mat{H}_{ii}{\mat{w}_i}}^H\cdot \overrightarrow{\mat{H}_{ki}{\mat{w}_k}}\big|^2 = \cos^2(\theta_{H,i})$ because of $\big\lVert\overrightarrow{\mat{H}_{ii}{\mat{w}_i}}\big\lVert^2 = \big\lVert\overrightarrow{\mat{H}_{ki}{\mat{w}_k}}\big\lVert^2 = 1$. Thus, (\ref{eq:MMSE SINR2}) becomes
\begin{align*}
&\mathrm{SINR}_i(\mat{w}_1,\mat{w}_{2}) \nonumber \\
&= \sin^2(\theta_{H,i})\cdot\frac{{\lVert \mat{H}_{ii}\mat{w}_i \lVert}^2}{\sigma_i^2} + \cos^2(\theta_{H,i})\cdot \frac{{\lVert \mat{H}_{ii}\mat{w}_i \lVert}^2}{\sigma_i^2 + {\lVert \mat{H}_{ki}\mat{w}_k \lVert}^2},\nonumber
\end{align*}
where $\theta_{H,i}\in[0,\pi/2]$ denotes the Hermitian angle between the desired signal direction ${\overrightarrow{\mat{H}_{ii}\mat{w}_i}}$ and the interference direction ${\overrightarrow{\mat{H}_{ki}\mat{w}_k}}$ at $\mathrm{RX}_i$. Obviously, when ${\overrightarrow{\mat{H}_{ii}\mat{w}_i}}\parallel{\overrightarrow{\mat{H}_{ki}\mat{w}_k}}$ (or ${\overrightarrow{\mat{H}_{ii}\mat{w}_i}}\perp{\overrightarrow{\mat{H}_{ki}\mat{w}_k}}$), we have $\theta_{H,i}=0$ (or $\theta_{H,i}={\pi}/{2}$).
\end{IEEEproof}

\section{Proof of Proposition 2}\label{sec:proofProp2}
\begin{IEEEproof} {Our idea is to show that it is impossible for a strict Pareto-optimal point achieved by the transmit beamformers with less than full power. The proof works by contradiction.}

{Assume that a strict Pareto-optimal point $\big(R_1(\mat{w}_1, \mat{w}_2), R_2(\mat{w}_1, \mat{w}_2)\big)$ is achieved by $(\mat{w}_1, \mat{w}_2)$ where $\lVert \mat{w}_1 \lVert^2<1$ and $\lVert {\mat{w}}_2 \lVert^2\leq1$. We consider whether there exists an \emph{outer}{\footnote{A point $\mat{p}'\in\mathbb{R}_n^+$ is called an outer point than $\mat{p}\in\mathbb{R}_n^+$, if $\mat{p}'$ dominates $\mat{p}$, i.e., $\mat{p}' \geq \mat{p}$ and $\mat{p}' \neq \mat{p}$ where the inequality is component-wise. The improvement from $\mat{p}$ to $\mat{p}'$ is called Pareto improvement.}} point $\big(R_1(\hat{\mat{w}}_1, \mat{w}_2), R_2(\hat{\mat{w}}_1, \mat{w}_2)\big)$ achieved by $(\hat{\mat{w}}_1, \mat{w}_2)$ where $\lVert \mat{w}_1 \lVert^2 < \lVert \hat{\mat{w}}_1 \lVert^2\leq1$ and $\lVert {\mat{w}}_2 \lVert^2\leq1$. If it exists, e.g., ${R}_1(\hat{\mat{w}}_1, \mat{w}_2)>{R}_1({\mat{w}}_1, \mat{w}_2)$ and ${R}_2(\hat{\mat{w}}_1, \mat{w}_2)=R_2(\mat{w}_1, \mat{w}_2)$, we can improve the ${R}_1({\mat{w}}_1, \mat{w}_2)$ only by consuming more transmit power while keeping ${R}_2({\mat{w}}_1, \mat{w}_2)$ unchanged. Thus, the existence of an \emph{outer} operating point contradicts the assumption that $\big(R_1(\mat{w}_1, \mat{w}_2), R_2(\mat{w}_1, \mat{w}_2)\big)$ is a strict Pareto-optimal point. }

Define $\hat{\mat{w}}_1={\mat{w}}_1+\mat{\delta}_{p}$. To guarantee the Pareto improvement, we need to show the existence of a
nonzero perturbation vector $\mat{\delta}_{p}$ satisfying:
\begin{subequations}\label{eq:ConFullPower}
\begin{align}
&(\mat{w}_1+\mat{\delta}_{p})^H\mat{A}_1(\mat{w}_2)(\mat{w}_1+\mat{\delta}_{p}) > \mat{w}_1^H\mat{A}_1(\mat{w}_2)\mat{w}_1\\
&\mat{w}_2^H\mat{A}_2(\mat{w}_1+\mat{\delta}_{p})\mat{w}_2 = \mat{w}_2^H\mat{A}_2(\mat{w}_1)\mat{w}_2\\
&\lVert \mat{w}_1+\mat{\delta}_{p}\lVert ^2>\lVert \mat{w}_1\lVert ^2\\
&\lVert \mat{w}_1+\mat{\delta}_{p}\lVert ^2\leq1.
\end{align}
\end{subequations}

An arbitrary nonzero $\mat{\delta}_{p}$ can be expressed as
\begin{align}\label{eq:deltaP}
\mat{\delta}_{p}=\lVert \mat{\delta}_{p} \lVert \cdot e^{j\phi_{\delta}} \cdot \overrightarrow{\mat{\delta}_{p}}.
\end{align}
It means that we should find $\lVert \mat{\delta}_{p} \lVert$, $\phi_{\delta}$ and $\overrightarrow{\mat{\delta}_{p}}$ to satisfy all the conditions in (\ref{eq:ConFullPower}) simultaneously. The proof of the existence of $\phi_{\delta}$ is similar to that for the two-user MISO IC in \cite{VS}. However, it is more difficult to find a $\overrightarrow{\mat{\delta}_{p}}$ for the MIMO IC because the cross-talk channel matrix (rather than a vector in the MISO IC case) does not always have a null space for $\overrightarrow{\mat{\delta}_{p}}$. We give the proof in detail as follows.

{\textbf{1. Existence of $\overrightarrow{\mat{\delta}_{p}}$}}

By the matrix inverse lemma \cite{MatrixAnalysisand}, the condition in (\ref{eq:ConFullPower}b) is equivalent to
\begin{align}\label{eq:ThetaF}
\frac{|(\mat{w}_1 + \mat{\delta}_{p})^H\mat{H}_{12}^H\mat{H}_{22}\mat{w}_2|^2}{\sigma_2^2+{\lVert \mat{H}_{12}(\mat{w}_1+ \mat{\delta}_{p})\lVert}^2} = \frac{|\mat{w}_1^H\mat{H}_{12}^H\mat{H}_{22}\mat{w}_2|^2}{\sigma_2^2+{\lVert \mat{H}_{12}\mat{w}_1\lVert}^2}.
\end{align}
It is difficult to solve $\mat{\delta}_{p}$ directly. In fact, we only need to prove the existence of $\mat{\delta}_{p}$ satisfying (\ref{eq:ThetaF}).

\textbf{Case~1.} {$N_R<N_T$ or $\mathrm{Rank}(\mat{H}_{12})<N_T\leq N_R$:

We always have $\mat{H}_{12}\mat{\delta}_p = \mat{0}$ if
\begin{align}\label{eq:deltaPerturb}
\overrightarrow{\mat{\delta}_{p}} = \sum_{i=1}^{\mathrm{Rank}(\Pi_{\mat{H}_{12}^T}^{\perp})}{a_i\mat{u}_i(\Pi_{\mat{H}_{12}^T}^{\perp})},
\end{align}
where $a_i, i=1,...,\mathrm{Rank}(\Pi_{\mat{H}_{12}^T}^{\perp})$ are complex-valued numbers and $\sum_{i=1}^{\mathrm{Rank}(\Pi_{\mat{H}_{12}^T}^{\perp})}{|a_i|^2}=1$.} Then, (\ref{eq:ThetaF}) always holds because any $\mat{\delta}_p$ in the null space of the cross-talk channel $\mat{H}_{12}$ does not cause extra interference to $\mathrm{RX}_2$.

\textbf{Case~2.} {$\mathrm{Rank}(\mat{H}_{12})=N_T\leq N_R$}:

It is impossible to nullify the perturbation directly as Case 1. Define $\mat{v}_1 \stackrel{\Delta}{=} \mat{H}_{12}\mat{w}_1$, $\mat{v}_{\delta} \stackrel{\Delta}{=} \mat{H}_{12}\mat{\delta}_{p}$ and $\mat{v}_2 \stackrel{\Delta}{=} \mat{H}_{22}\mat{w}_2$. Then (\ref{eq:ThetaF}) becomes
\begin{align}\label{eq:ThetaFV}
\frac{|(\mat{v}_1 + \mat{v}_{\delta})^H\mat{v}_2|^2}{\sigma_2^2+{\lVert \mat{v}_1 + \mat{v}_{\delta} \lVert}^2} = \frac{|\mat{v}_1^H\mat{v}_2|^2}{\sigma_2^2+{\lVert \mat{v}_1 \lVert}^2}.
\end{align}

Assume that $\mat{v}_{\delta}$ is a combination of two orthogonal vectors
{
\begin{align}\label{eq:Vdelta}
\mat{v}_{\delta} \stackrel{\Delta}{=} \lVert \mat{v}_{\delta}\lVert \cdot \big(\sqrt{\eta} \cdot\overrightarrow{\Pi_{\mat{v}_2}^{\perp}{\mat{v}_1}} + \sqrt{1-\eta}\cdot\overrightarrow{\Pi_{\mat{v}_2}{\mat{v}_1}}\big),
\end{align}
Note that $\overrightarrow{\Pi_{\mat{v}_2}{\mat{v}_1}} = \overrightarrow{\mat{v}_2} \cdot e^{-j\phi_{1}}$ where $\phi_{1}=\arg({\mat{v}_1^H\mat{v}_2})$.} Now, it {remains} to find whether there is a $\mat{v}_{\delta}$ in the plane spanned by $\overrightarrow{\Pi_{\mat{v}_2}^{\perp}{\mat{v}_1}}$ and $\overrightarrow{\Pi_{\mat{v}_2}{\mat{v}_1}}$ satisfying (\ref{eq:ThetaFV}).

Substituting (\ref{eq:Vdelta}) into (\ref{eq:ThetaFV}) yields
\begin{align}\label{eq:FVcombination}
&{\frac{\big |\mat{v}_1^H\mat{v}_2 + {\lVert \mat{v}_{\delta}\lVert} \cdot \sqrt{1-\eta}\cdot e^{j\phi_{1}}\cdot \lVert \mat{v}_2 \lVert\big|^2}{\sigma_2^2+{\left\lVert \mat{v}_1 + {\lVert \mat{v}_{\delta}\lVert} \big(\sqrt{\eta}\cdot\overrightarrow{\Pi_{\mat{v}_2}^{\perp}{\mat{v}_1}} + \sqrt{1-\eta}\cdot e^{-j\phi_{1}}\cdot{\overrightarrow{\mat{v}_2}}\big) \right\lVert}^2}} \nonumber \\
&~~= {\frac{|\mat{v}_1^H\mat{v}_2|^2}{\sigma_2^2+{\lVert \mat{v}_1 \lVert}^2}}.
\end{align}
Define the right-hand side and the left-hand side of (\ref{eq:FVcombination}) as $R_{side}$ and $L_{side}(\eta)$, respectively. It is still hard to get a closed-form solution of $\eta$ by solving (\ref{eq:FVcombination}) directly. Observe that the denominator of $L_{side}(\eta)$ is always positive for $\eta\in[0,1]$, and $L_{side}(\eta)$ as a function of $\eta$ is continuous over the interval $[0,1]$. Therefore, if $(R_{side}-L_{side}(1))(R_{side}-L_{side}(0))\leq0$, there must exist a ${\mat{v}_{\delta}(\eta)}$ with at least a certain $\eta\in[0,1]$ satisfying (\ref{eq:FVcombination}).

When $\eta=1$, the term $L_{side}(\eta)$ becomes
\begin{equation*}\label{eq:FVcombination1}
\begin{aligned}
L_{side}(1) &= \frac{|\mat{v}_1^H\mat{v}_2|^2}{\sigma_2^2+{\lVert \mat{v}_1 + {\lVert \mat{v}_{\delta}\lVert} \cdot\overrightarrow{\Pi_{\mat{v}_2}^{\perp}{\mat{v}_1}}\lVert}^2} \nonumber \\
&= \frac{|\mat{v}_1^H\mat{v}_2|^2}{\sigma_2^2+{\lVert \mat{v}_1 \lVert}^2 + {\lVert \mat{v}_{\delta}\lVert^2} + 2{\lVert \mat{v}_{\delta}\lVert} \cdot \lVert \Pi_{\mat{v}_2}^{\perp}{\mat{v}_1}\lVert}.\nonumber
\end{aligned}
\end{equation*}
Observe that ${\lVert \mat{v}_{\delta}\lVert^2} + 2{\lVert \mat{v}_{\delta}\lVert} \cdot \lVert \Pi_{\mat{v}_2}^{\perp}{\mat{v}_1}\lVert >0$. Thus, we have $L_{side}(1)<R_{side}$.

When $\eta=0$, the term $L_{side}(\eta)$ becomes
\begin{equation*}\label{eq:FVcombination0}
\begin{aligned}
L_{side}(0) &= \frac{\big|\mat{v}_1^H\mat{v}_2 + {\lVert \mat{v}_{\delta}\lVert}\cdot e^{j\phi_{1}} \cdot\lVert \mat{v}_2 \lVert\big|^2}{\sigma_2^2+{\lVert \mat{v}_1 + {\lVert \mat{v}_{\delta}\lVert} \cdot e^{-j\phi_{1}} \cdot \overrightarrow{\mat{v}_2}\lVert}^2} \nonumber \\
&= \frac{|\mat{v}_1^H\mat{v}_2|^2 + \lVert\mat{v}_2 \lVert^2 \cdot({\lVert \mat{v}_{\delta}\lVert^2} + 2 {\lVert \mat{v}_{\delta}\lVert}\cdot |\mat{v}_1^H\overrightarrow{\mat{v}_2}|)}{\sigma_2^2+\lVert \mat{v}_1\lVert^2 + ({\lVert \mat{v}_{\delta}\lVert^2} + 2{\lVert \mat{v}_{\delta}\lVert}\cdot|\mat{v}_1^H \overrightarrow{\mat{v}_2}|)}, \nonumber
\end{aligned}
\end{equation*}
where ${{\lVert \mat{v}_{\delta}\lVert^2}} + {2}\cdot {\lVert \mat{v}_{\delta}\lVert} \cdot |\mat{v}_1^H \overrightarrow{\mat{v}_2}|>0$. If $L_{side}(0)> R_{side}$, we need $\lVert \mat{v}_2 \lVert^2 > R_{side}$. Furthermore, $R_{side}$ is bounded by
\begin{align}\label{eq:Leftside}
R_{side} = {\frac{|\mat{v}_1^H\mat{v}_2|^2}{\sigma_2^2+{\lVert \mat{v}_1 \lVert}^2}} \leq {\frac{\lVert \mat{v}_1\lVert^2\cdot \lVert \mat{v}_2\lVert^2}{\sigma_2^2+{\lVert \mat{v}_1 \lVert}^2}}
= \frac{\lVert \mat{v}_2\lVert^2}{\frac{\sigma_2^2}{\lVert \mat{v}_1\lVert^2}+1} < {\lVert \mat{v}_2\lVert}^2. \nonumber
\end{align}
Thus, we have $L_{side}(0)> R_{side}$.

Due to $L_{side}(1)< R_{side} < L_{side}(0)$, there exists at least one $\eta_0\in(0,1)$ satisfying $L_{side}(\eta_0)= R_{side}$. In this case, {$\mat{H}_{12}$ has an inverse/Moore-Penrose pseudo-inverse matrix $\mat{H}_{12}^{\dagger}$.}
Then, we have
\begin{equation}\label{eq:fullrankexpress}
\begin{aligned}
\mat{\delta}_p &= {\lVert \mat{v}_{\delta}\lVert \cdot\mat{H}_{12}^{\dagger}\big(\sqrt{\eta_0}\cdot\overrightarrow{\Pi_{\mat{v}_2}^{\perp}{\mat{v}_1}} + \sqrt{1-\eta_0} \cdot\overrightarrow{\Pi_{\mat{v}_2}{\mat{v}_1}}\big)},\\
&= \lVert\mat{\delta}_{p}\lVert\cdot e^{j\phi_{\delta}}\cdot \overrightarrow{\mat{\delta}_p},
\end{aligned}
\end{equation}
{where $\lVert \mat{v}_{\delta}\lVert = \frac{\lVert\mat{\delta}_{p}\lVert}{\big\lVert \mat{H}_{12}^{\dagger}\big(\sqrt{\eta_0}\cdot\overrightarrow{\Pi_{\mat{v}_2}^{\perp}{\mat{v}_1}} + \sqrt{1-\eta_0} \cdot\overrightarrow{\Pi_{\mat{v}_2}{\mat{v}_1}}\big)\big \lVert}$ depends on but has no requirement for $\lVert\mat{\delta}_{p}\lVert$. Therefore, $\mat{v}_{\delta}$ with any $\lVert\mat{\delta}_{p}\lVert$ and $\overrightarrow{\mat{\delta}_p}=  e^{-j\phi_{\delta}}\cdot \overrightarrow{\mat{H}_{12}^{\dagger}\big(\sqrt{\eta_0} \cdot\overrightarrow{\Pi_{\mat{v}_2}^{\perp}{\mat{v}_1}} + \sqrt{1-\eta_0}\cdot\overrightarrow{\Pi_{\mat{v}_2}{\mat{v}_1}}\big)}$ satisfies (\ref{eq:ThetaFV}).}

Therefore, any $\mat{\delta}_p = \lVert\mat{\delta}_{p}\lVert\cdot e^{j\phi_{\delta}}\cdot \overrightarrow{\mat{\delta}_p}$ with
\begin{equation}\label{eq:allrankdirection}
\overrightarrow{\mat{\delta}_{p}} =
\left\{
\begin{aligned}
&\sum_{i=1}^{\mathrm{Rank}(\Pi_{\mat{H}_{12}^T}^{\perp})}{a_i\mat{u}_i(\Pi_{\mat{H}_{12}^T}^{\perp})},~~ \sum_{i=1}^{\mathrm{Rank}(\Pi_{\mat{H}_{12}^T}^{\perp})}{|a_i|^2}=1, \\
&~~~~\mathrm{when} ~N_R<N_T~ \mathrm{or}~ \mathrm{Rank}(\mat{H}_{12})<N_T\leq N_R \\
&e^{-j\phi_{\delta}}\cdot \overrightarrow{\mat{H}_{12}^{\dagger}\big(\sqrt{\eta_0} \cdot\overrightarrow{\Pi_{\mat{v}_2}^{\perp}{\mat{v}_1}} + \sqrt{1-\eta_0}\cdot\overrightarrow{\Pi_{\mat{v}_2}{\mat{v}_1}}\big)}, \\
&~~~~\mathrm{when} ~\mathrm{Rank}(\mat{H}_{12})=N_T\leq N_R,
\end{aligned}
\right.
\end{equation}
satisfies the condition (\ref{eq:ConFullPower}b).

{\textbf{2. Existence of $\phi_{\delta}$}}

Substituting (\ref{eq:deltaP}) into (\ref{eq:ConFullPower}a) yields
\begin{align}\label{eq:ConFullPowera}
&(\mat{w}_1+\mat{\delta}_{p})^H\mat{A}_1(\mat{w}_2)(\mat{w}_1+\mat{\delta}_{p}) > \mat{w}_1^H\mat{A}_1(\mat{w}_2)\mat{w}_1 \nonumber \\
\Leftrightarrow &\lVert\mat{\delta}_{p}\lVert^2\overrightarrow{\mat{\delta}_{p}}^H\mat{A}_1(\mat{w}_2)\overrightarrow{\mat{\delta}_{p}}
+ 2\lVert\mat{\delta}_{p}\lVert\Re\Big(\mat{w}_1^H \mat{A}_1(\mat{w}_2) \overrightarrow{\mat{\delta}_{p}}{e^{j\phi_{\delta}}}\Big)>0 \nonumber \\
\Leftrightarrow
&\frac{\lVert\mat{\delta}_{p}\lVert}{2}\overrightarrow{\mat{\delta}_{p}}^H \mat{A}_1(\mat{w}_2)\overrightarrow{\mat{\delta}_{p}} +  \big |\mat{w}_1^H \mat{A}_1(\mat{w}_2) \overrightarrow{\mat{\delta}_{p}}\big|\cos(\phi_{\delta}+\phi_{2})>0 \nonumber \\
\Leftrightarrow
&\cos(\phi_\delta+\phi_{2}) > -\frac{\lVert\mat{\delta}_{p}\lVert}{2} \cdot \frac{\overrightarrow{\mat{\delta}_{p}}^H\mat{A}_1(\mat{w}_2)\overrightarrow{\mat{\delta}_{p}}}{\big|\mat{w}_1^H \mat{A}_1(\mat{w}_2) \overrightarrow{\mat{\delta}_{p}}\big|},
\end{align}
where {$\phi_{2} \stackrel{\Delta}{=} \arg(\mat{w}_1^H \mat{A}_1(\mat{w}_2) \overrightarrow{\mat{\delta}_{p}})$}.

At the same time, substituting (\ref{eq:deltaP}) into (\ref{eq:ConFullPower}c) yields
\begin{align}\label{eq:ConFullPowerc}
&\lVert \mat{w}_1+\mat{\delta}_{p}\lVert ^2>\lVert \mat{w}_1\lVert ^2 \nonumber \\
\Longleftrightarrow
~&\lVert\mat{\delta}_{p}\lVert^2 + 2\lVert\mat{\delta}_{p}\lVert\Re\Big(\mat{w}_1^H\overrightarrow{\mat{\delta}_p} {e^{j\phi_{\delta}}}\Big)>0 \nonumber \\
\Longleftrightarrow
~&\big|\mat{w}_1^H\overrightarrow{\mat{\delta}_p}\big|\cos(\phi_{\delta}+\phi_{3})>-\frac{\lVert\mat{\delta}_{p}\lVert}{2} \nonumber \\
\Longleftrightarrow ~&\cos(\phi_{\delta}+\phi_{3})>-\frac{\lVert\mat{\delta}_{p}\lVert}{2|\mat{w}_1^H\overrightarrow{\mat{\delta}_p}|},
\end{align}
where $\phi_{3} \stackrel{\Delta}{=} \arg(\mat{w}_1^H\overrightarrow{\mat{\delta}_p})$.
{

Define $\phi_{\delta}+\phi_{2} \in [\theta_1, \theta_2]$ and $\phi_{\delta}+\phi_{3} \in [\theta_3, \theta_4]$. Since both the right-hand side of (\ref{eq:ConFullPowera}) and (\ref{eq:ConFullPowerc}) are negative, the range $[\theta_1, \theta_2]$ and $[\theta_3, \theta_4]$ are strictly wider than $\pi$. In addition, the intersection of two angular ranges wider than $\pi$ is nonempty.} Then, for arbitrary $\lVert \mat{\delta}_p \lVert$ and $\overrightarrow{\mat{\delta}_p}$, any $\mat{\delta}_p = \lVert \mat{\delta}_p \lVert \cdot e^{j\phi_{\delta}} \overrightarrow{\mat{\delta}_p}$ in (\ref{eq:deltaP}) with $\phi_{\delta}\in[\theta_1, \theta_2]\cap[\theta_3, \theta_4]$ always satisfies the conditions (\ref{eq:ConFullPowera}) and (\ref{eq:ConFullPowerc}) simultaneously.

{\textbf{3. Existence of $\lVert\mat{\delta}_{p}\lVert$}}

The condition (\ref{eq:ConFullPower}d) is equivalent to
\begin{align}\label{eq:ConFullPowerd}
&\lVert \mat{\delta}_p \lVert^2 + 2|\mat{w}_1^H\overrightarrow{\mat{\delta}_p}|\cos(\phi_{\delta}+\phi_{3}) \lVert \mat{\delta}_p \lVert + \lVert \mat{w}_1 \lVert^2 -1 \leq 0, \nonumber\\
\stackrel{(a)}{\Longleftrightarrow}~
& \lVert \mat{\delta}_p \lVert \in \Big(0, -|\mat{w}_1^H\overrightarrow{\mat{\delta}_p}|\cos(\phi_{\delta}+\phi_{3}) \nonumber \\
&~~+\sqrt{|\mat{w}_1^H\overrightarrow{\mat{\delta}_p}|^2\cos^2(\phi_{\delta}+\phi_{3})-(\lVert \mat{w}_1\lVert^2 -1)} \Big),
\end{align}
where the transformation (a) is based on $\lVert \mat{w}_1\lVert\leq 1$ and $\lVert \mat{\delta}_p \lVert>0$. For arbitrary $\overrightarrow{\mat{\delta}_p}$ and $\phi_{\delta}$, any any $\mat{\delta}_p = \lVert \mat{\delta}_p \lVert \cdot e^{j\phi_{\delta}} \overrightarrow{\mat{\delta}_p}$ with $\lVert \mat{\delta}_p \lVert$ in (\ref{eq:ConFullPowerd}) will satisfy the condition (\ref{eq:ConFullPower}d) .

{Above all, the existence of $\overrightarrow{\mat{\delta}_{p}}$, $\phi_{\delta}$ and $\lVert\mat{\delta}_{p}\lVert$ has been proved.} That is, there always exists some $\mat{\delta}_p = \lVert \mat{\delta}_p \lVert\cdot e^{j\phi_\delta}\cdot\overrightarrow{\mat{\delta}_p}$ satisfying all the conditions in (\ref{eq:ConFullPower}). {Then, $R_1(\mat{w}_1, \mat{w}_2)$ can still be improved until $\lVert \mat{w}_1 \lVert^2 =1$, while $R_1(\mat{w}_1, \mat{w}_2)$ remains unchanged simultaneously. This contradicts the assumption that $\big(R_1(\mat{w}_1, \mat{w}_2), R_2(\mat{w}_1, \mat{w}_2)\big)$ is on the strict Pareto boundary. Therefore, Proposition 2 holds.}
\end{IEEEproof}

\section{Proof of Proposition 3}\label{sec:proofProp3}
\begin{IEEEproof}
{For the ending point $E1(\overline{R}_1, \underline{R}_2)$, to achieve the maximum rate of link $\mathrm{TX}_1 \mapsto \mathrm{RX}_1$, i.e., $\overline{R}_1$ in (\ref{eq:SURate}a), $(\mat{w}_1, \mat{w}_2)$ should satisfy the following conditions}
\begin{subequations}\label{eq:ConTP1}
\begin{align}
&\mat{w}_1 = \mat{w}_1^{Ego} = \mat{u}_1(\mat{H}_{11}^H\mat{H}_{11}),\\
&\theta_{H,1}=\pi/2 \Leftrightarrow {\mat{w}_2}\perp{\mat{H}_{21}^H\mat{H}_{11}\mat{w}_1}.
\end{align}
\end{subequations}
This means ${\mat{w}_2}$ should be in the null space of $\mat{H}_{21}^H\mat{H}_{11}\mat{w}_1^{Ego}$ to cause no interference to $\mathrm{RX}_1$. All $\mat{w}_2\in\mathcal{W_{FP}}$ satisfying (\ref{eq:ConTP1}) form a set of the ZF strategies, defined as $\mathcal{W_{ZF}}$. Then, any $\mat{w}_2\in\mathcal{W_{ZF}}$ can be expressed by
\begin{align}\label{eq:BeamTP2}
\mat{w}_2 = \overrightarrow{\Pi_{\mat{H}_{21}^H\mat{H}_{11}\mat{w}_1^{Ego}}^{\perp}\mat{v}_{2}},
\end{align}
where $\mat{v}_{2}\in\mathbb{C}^{N_T\times 1}$ and $\mat{v}_{2}\nparallel \mat{H}_{21}^H\mat{H}_{11}\mat{w}_1^{Ego}$.

{To achieve $(\overline{R}_1, \underline{R}_2)$, we need to find $\mat{v}_{2}^{opt}$ which maximizes $\mathrm{SINR}_2(\mat{w}_1^{Ego}, \mat{w}_2)$ simultaneously. Here, we define the optimal "altruistic" strategy} $\mat{w}_2^{Alt}$ as
\begin{align}\label{eq:OptiBeamAlt}
\mat{w}_2^{Alt} &\stackrel{\Delta}{=} \arg \max_{\mat{w}_2\in \mathcal{W_{ZF}}} {\mat{w}_{2}^H\mat{A}_2(\mat{w}_1^{Ego})\mat{w}_{2}} ~~\stackrel{(a)}{\Longleftrightarrow}\nonumber \\
\mat{v}_2^{opt} &= \arg \max_{\mat{v}_{2}} {\frac{\mat{v}_{2}^H { \Pi_{\mat{H}_{21}^H\mat{H}_{11}\mat{w}_1^{Ego}}^{\perp,H} \mat{A}_2(\mat{w}_1^{Ego}) \Pi_{\mat{H}_{21}^H\mat{H}_{11}\mat{w}_1^{Ego}}^{\perp} } \mat{v}_{2}} {\mat{v}_{2}^H {\Pi_{\mat{H}_{21}^H\mat{H}_{11}\mat{w}_1^{Ego}}^{\perp,H}\Pi_{\mat{H}_{21}^H\mat{H}_{11}\mat{w}_1^{Ego}}^{\perp}}\mat{v}_{2}}} \nonumber\\
\stackrel{(b)}{=}& \arg \max_{\mat{v}_{2}} {\frac{\mat{v}_{2}^H \overbrace{ \Pi_{\mat{H}_{21}^H\mat{H}_{11}\mat{w}_1^{Ego}}^{\perp} \mat{A}_2(\mat{w}_1^{Ego}) \Pi_{\mat{H}_{21}^H\mat{H}_{11}\mat{w}_1^{Ego}}^{\perp} }^{\stackrel{\Delta}{=}\mat{B}_{1}} \mat{v}_{2}} {\mat{v}_{2}^H {\Pi_{\mat{H}_{21}^H\mat{H}_{11}\mat{w}_1^{Ego}}^{\perp}}\mat{v}_{2}}} \nonumber\\
=& \mat{u}_1\big(\mat{B}_{1}, \Pi_{\mat{H}_{21}^H\mat{H}_{11}\mat{w}_1^{Ego}}^{\perp}\big),
\end{align}
where transformation (a) is based on (\ref{eq:BeamTP2}) and transformation (b) is due to ${\Pi_{\mat{H}_{21}^H\mat{H}_{11}\mat{w}_1^{Ego}}^{\perp,H}} = {\Pi_{\mat{H}_{21}^H\mat{H}_{11}\mat{w}_1^{Ego}}^{\perp}}$ and ${\Pi_{\mat{H}_{21}^H\mat{H}_{11}\mat{w}_1^{Ego}}^{\perp,H}}{\Pi_{\mat{H}_{21}^H\mat{H}_{11}\mat{w}_1^{Ego}}^{\perp}} = {\Pi_{\mat{H}_{21}^H\mat{H}_{11}\mat{w}_1^{Ego}}^{\perp}}$.
{Substituting (\ref{eq:OptiBeamAlt}) into (\ref{eq:BeamTP2}), we obtain the optimal "altruistic" strategy}
\begin{align}
\mat{w}_2^{Alt} = \overrightarrow{\Pi_{\mat{H}_{21}^H\mat{H}_{11}\mat{w}_1^{Ego}}^{\perp}\mat{u}_1(\mat{B}_{1}, \Pi_{\mat{H}_{21}^H\mat{H}_{11}\mat{w}_1^{Ego}}^{\perp})}.\nonumber
\end{align}

Therefore, $E1(\overline{R}_1, \underline{R}_2)$ is achieved by
$(\mat{w}_1^{Ego}, \mat{w}_1^{Alt})$.
\end{IEEEproof}

\section{Solving the SDR problem $\mathrm{(P4)}$}\label{sec:SolveP4}

\begin{lemma}
Both the fractional problem $\mathrm{(P4)}$ and the problem $\mathrm{(P5)}$ are solvable.
\end{lemma}
\begin{IEEEproof}
For the fractional problem $\mathrm{(P4)}$, its constraint set
\begin{align}
\Omega = \{\mat{W}_2\succeq\mat{0} : \mathrm{Tr}\big(\mat{W}_2\big) = 1,~ \mathrm{Tr}\big(\mat{A}_2(\mat{w}_1)\mat{W}_2\big)  = \mathrm{SINR}_2^\star \} \nonumber
\end{align}
is nonempty and compact.

In $\mathrm{(P4)}$, the denominator of the objective over $\Omega$ satisfies
\begin{align}\label{eq:P5denominator}
&\sigma_1^2 + \lambda_{N_T}(\mat{H}_{21}^H\mat{H}_{21}) = \lambda_{N_T}(\mat{C}_{2}) \leq \mathrm{Tr}\big(\mat{C}_{2}\mat{W}_2\big) \nonumber \\
&~~~~~~~~~~\leq \lambda_1(\mat{C}_{2}) = \sigma_1^2 + \lambda_1(\mat{H}_{21}^H\mat{H}_{21}),\nonumber
\end{align}
and the numerator obviously satisfies
\begin{align}\label{eq:P5numerator}
0 \leq {\mathrm{Tr}\big(\mat{C}_{1}(\mat{w}_1)\mat{W}_2\big)} \leq \lambda_1(\mat{H}_{21}^H\mat{H}_{11}\mat{H}_{11}^H\mat{H}_{21}). \nonumber
\end{align}
This implies that the objective $\mathrm{(P4)}$ over $\Omega$ is bounded by
\begin{align}\label{eq:P5bound}
0 \leq \frac{{\mathrm{Tr}\big(\mat{C}_{1}(\mat{w}_1)\mat{W}_2\big)}}{\mathrm{Tr}\big(\mat{C}_{2}\mat{W}_2\big)} \leq \frac{\lambda_1(\mat{H}_{21}^H\mat{H}_{11}\mat{H}_{11}^H\mat{H}_{21})}{\sigma_1^2 + \lambda_{N_T}(\mat{H}_{21}^H\mat{H}_{21})}.
\end{align}

Based on Weierstrass' Theorem, the problem $\mathrm{(P4)}$ always has an optimal solution.

Assume that $\mat{W}_2^\star$ is an optimal solution to $\mathrm{(P4)}$, we know that $s^\star=\frac{1}{\mathrm{Tr}\big(\mat{C}_{2}\mat{W}_2^\star\big)}$ and $\mat{Q}^\star = s^\star \mat{W}_2^\star$ are feasible for $\mathrm{(P5)}$. Also note that the objective is bounded by (\ref{eq:P5bound}). Similarly, $\mathrm{(P5)}$ is solvable according to Weierstrass' Theorem.
\end{IEEEproof}

\begin{lemma}
The problems $\mathrm{(P4)}$ and $\mathrm{(P5)}$ have the same value. Furthermore, if $\mat{W}_2^\star$ solves {$\mathrm{(P4)}$}, then $s^\star=\frac{1}{\mathrm{Tr}\big({\mat{C}_{2}}\mat{W}_2^\star\big)}$ and $\mat{Q}^\star = s^\star \cdot \mat{W}_2^\star$ solves $\mathrm{(P5)}$; if $\mat{Q}^\star$ and $s^\star$ solve $\mathrm{(P5)}$, then ${\mat{W}_2^\star=}\frac{\mat{Q}^\star}{s^\star}$ solves {$\mathrm{(P4)}$}.
\end{lemma}
\begin{IEEEproof}
Assume that $\mat{W}_2^\star$ is an optimal solution to $\mathrm{(P4)}$, and $v_{\mathrm{(P4)}}^\star$ and $v_{\mathrm{(P5)}}^\star$ are the optimal values of the objective of $\mathrm{(P4)}$ and $\mathrm{(P5)}$, respectively. Thus, $s=\frac{1}{\mathrm{Tr}\big(\mat{C}_{2}\mat{W}_2^\star\big)}$ and $\mat{Q} = s \mat{W}_2^\star$ are feasible for ${\mathrm{(P5)}}$. The value of the objective of ${\mathrm{(P5)}}$ at this feasible point is
\begin{align}\label{eq:P5P6Relation1}
v_{\mathrm{(P5)}} &= {\mathrm{Tr}\big(\mat{C}_{1}(\mat{w}_1)\mat{Q}\big)} \nonumber \\
&= {\mathrm{Tr}\big(\mat{C}_{1}(\mat{w}_1)(s\cdot \mat{W}_2^\star)\big)} = \frac{{\mathrm{Tr}\big(\mat{C}_{1}(\mat{w}_1)\mat{W}_2^\star\big)}}{\mathrm{Tr}\big(\mat{C}_{2}\mat{W}_2^\star\big)}= v_{\mathrm{(P4)}}^\star \nonumber \\
& \geq v_{\mathrm{(P5)}}^\star. \nonumber
\end{align}

On the other hand, suppose that $\mat{Q}^\star$ and $s^\star$ are the optimal solutions to $\mathrm{(P5)}$. Since $s^\star$ is always positive, $\mat{W}_2 = \frac{\mat{Q}^\star}{s^\star}$ is also feasible for $\mathrm{(P4)}$. Then, the value of the objective of $\mathrm{(P4)}$ at this feasible point is
\begin{align}\label{eq:P5P6Relation2}
v_{\mathrm{(P4)}} &= \frac{\mathrm{Tr}\big(\mat{C}_{1}(\mat{w}_1)\mat{W}_2\big)}{\mathrm{Tr}\big(\mat{C}_{2}\mat{W}_2\big)}= \frac{{\mathrm{Tr}\big(\mat{C}_{1}(\mat{w}_1)\frac{\mat{Q}^\star}{s^\star}\big)}}{\mathrm{Tr}\big(\mat{C}_{2}\frac{\mat{Q}^\star}{s^\star}\big)} \nonumber \\
& = \frac{{\mathrm{Tr}\big(\mat{C}_{1}(\mat{w}_1)\mat{Q}^\star\big)}}{\mathrm{Tr}\big(\mat{C}_{2}\mat{Q}^\star\big)} = {{\mathrm{Tr}\big(\mat{C}_{1}(\mat{w}_1)\mat{Q}^\star\big)}}= v_{\mathrm{(P5)}}^\star \nonumber \\
& \geq v_{\mathrm{(P4)}}^\star. \nonumber
\end{align}

Above all, we have $v_{\mathrm{(P4)}}^\star = v_{\mathrm{(P5)}}^\star$.
\end{IEEEproof}

\section{Proof of Proposition 4}\label{sec:proofProp4}
\begin{IEEEproof}
We need to find a feasible set $\mathcal{W_{F}}$ such that there exists at least one solution $\mat{w}_1\in\mathcal{W_{FP}}$ to problem $\mathrm{(P0)}$ by fixing $\mat{w}_2\in\mathcal{W_{F}}$.

In (\ref{eq:FractionConstriant}b), we derive that the constraint of $\mathrm{(P0)}$ is equivalent to $\mat{w}_2^H\mat{H}_{22}^H\mat{H}_{22}\mat{w}_2 \geq \sigma_2^2\mathrm{SINR}_2^\star$ and
\begin{align}\label{eq:FractionConstriantC}
\mat{w}_1^H\mat{C}(\mat{w}_2)\mat{w}_1 = 0.
\end{align}

To guarantee the existence of $\mat{w}_1\in\mathcal{W_{FP}}$ in (\ref{eq:FractionConstriantC}), a feasible $\mat{w}_2$ should be determined in a way such that (\ref{eq:FractionConstriant}b) holds.

By the eigen-decomposition, $\mat{C}(\mat{w}_2)$ can be rewritten as $\sum_{i=1}^{N_T}{\lambda_i(\mat{C}(\mat{w}_2))\mat{u}_i(\mat{C}(\mat{w}_2))\mat{u}_i^H(\mat{C}(\mat{w}_2))}$. We analyze $\mat{C}(\mat{w}_2)$ for two cases.

\textbf{Case~1.} When $\mat{C}(\mat{w}_2)$ is a full rank matrix, i.e., $\lambda_i(\mat{C}(\mat{w}_2)) \neq 0, \forall i=1,...,N_T$.
If $\mat{C}(\mat{w}_2)$ is a positive or negative definite matrix, it is clear that there is no nonzero vector $\mat{w}_1$ satisfying (\ref{eq:FractionConstriantC}). Otherwise, $\mat{C}(\mat{w}_2)$ has $\lambda_1(\mat{C}(\mat{w}_2))>0$ and ${\lambda_{N_T}(\mat{C}(\mat{w}_2))<0}$, a sufficient solution to (\ref{eq:FractionConstriantC}) is
\begin{align}
\mat{w}_1 &= \sqrt{\frac{\lambda_{N_T}(\mat{C}(\mat{w}_2))}{\lambda_1(\mat{C}(\mat{w}_2))-\lambda_{N_T}(\mat{C}(\mat{w}_2))}} j \cdot \mat{u}_1(\mat{C}(\mat{w}_2)) \nonumber \\
&~~+ \sqrt{\frac{\lambda_1(\mat{C}(\mat{w}_2))}{\lambda_1(\mat{C}(\mat{w}_2))-\lambda_{N_T}(\mat{C}(\mat{w}_2))}} \cdot \mat{u}_{N_T}(\mat{C}(\mat{w}_2)). \nonumber
\end{align}

\textbf{Case~2.} When $\mat{C}(\mat{w}_2)$ is not a full rank matrix, i.e., $\lambda_i(\mat{C}(\mat{w}_2))= 0$ for some $i \in\{1,...,N_T\}$.
{In this case, $\mat{C}(\mat{w}_2)$ always has null space for $\mat{w}_1$ to satisfy (\ref{eq:FractionConstriantC}). }

Above all, the sufficient and necessary condition of $\mat{w}_2$ satisfying (\ref{eq:FractionConstriantC}) is $\lambda_1\left(\mat{C}(\mat{w}_2)\right)\cdot \lambda_{N_T}\left(\mat{C}(\mat{w}_2)\right) \leq 0$. That is, any $\mat{w}_2 \in \mathcal{W_F}$ is always feasible for (\ref{eq:FractionConstriant}b) where $\mathcal{W_F}$ is
\begin{align}
\mathcal{W_F}\stackrel{\Delta}{=} &\Big\{\mat{w}_2\in \mathcal{W_{FP}} :\mat{w}_2^H\mat{H}_{22}^H\mat{H}_{22}\mat{w}_2 \geq \sigma_2^2\mathrm{SINR}_2^\star, \nonumber \\
&~~~~~~~\lambda_1\left(\mat{C}(\mat{w}_2)\right)\cdot \lambda_{N_T}\left(\mat{C}(\mat{w}_2)\right) \leq 0\Big\}.\nonumber
\end{align}
\end{IEEEproof}

\bibliographystyle{IEEEbib}
\bibliography{ref3}

\begin{thebibliography}{10}

\bibitem{ParetoB}
L.~Zadeh,
\newblock ``{Optimality and non-scalar-valued performance criteria},''
\newblock {\em IEEE Trans. Automat. Contr.}, vol. 8, no. 1, pp. 59--60, Jan.
  1963.

\bibitem{EduardMISOCharcter}
E.~Jorswieck, E.~Larsson, and D.~Danev,
\newblock ``{Complete characterization of the Pareto boundary for the MISO
  interference channel},''
\newblock {\em IEEE Trans. Signal Process.}, vol. 56, no. 10, pp. 5292--5296,
  Oct. 2008.

\bibitem{Closed-FormMISO}
J.~Lindblom, E.~Karipidis, and E.~Larsson,
\newblock ``{Closed-form parameterization of the Pareto boundary for the
  two-user MISO interference channel},''
\newblock in {\em Proc. IEEE ICASSP}, Prague, Czech Republic, May 2011.

\bibitem{WalrasianEqRami}
R.~Mochaourab and E.~Jorswieck,
\newblock ``{Walrasian equilibrium in two-user multiple-input single-output
  interference channel},''
\newblock in {\em Proc. IEEE ICC}, Kyoto, Japan, 2011.

\bibitem{RamiTSP2011}
R.~Mochaourab and E.~Jorswieck,
\newblock ``Optimal beamforming in interference networks with perfect local
  channel information,''
\newblock {\em IEEE Trans. Signal Process.}, vol. 59, no. 3, pp. 1128 -- 1141,
  Mar. 2011.

\bibitem{CharacterizeMIMOSimplereceiver}
E.~Bjornson, M.~Bengtsson, and B.~Ottersten,
\newblock ``{Pareto characterization of the multicell MIMO performance region
  with simple receivers},''
\newblock {\em IEEE Trans. Signal Process.}, vol. 60, no. 8, pp. 4464--4469,
  Aug. 2012.

\bibitem{EfficientComputationMISO}
E.~Karipidis and E.~Larsson,
\newblock ``{Efficient computation of the Pareto boundary for the MISO
  interference channel with perfect CSI},''
\newblock in {\em Proc. WiOpt 2010}, Avignon, France, June 2010.

\bibitem{CaoSPAWC2012}
P.~Cao, S.~Shi, and E.~Jorswieck,
\newblock ``{Efficient computation of the Pareto boundary for the two-user
  single-stream MIMO interference channel},''
\newblock in {\em Proc. IEEE SPAWC}, Cesme, Turkey, June 2012.

\bibitem{RZhangMISO}
R.~Zhang and S.~Cui,
\newblock ``{Cooperative interference management with MISO beamforming},''
\newblock {\em IEEE Trans. Signal Process.}, vol. 58, no. 10, pp. 5450--5458,
  Oct. 2010.

\bibitem{RZhangMISO2}
J.~Qiu, R.~Zhang, Z.-Q. Luo, and S.~Cui,
\newblock ``{Optimal distributed beamforming for MISO interference channels},''
\newblock {\em IEEE Trans. Signal Process.}, vol. 59, no. 11, pp. 5638--5643,
  Nov. 2011.

\bibitem{MAPEL}
L.~Qian, A.~Y. Zhang, and J.~Huang,
\newblock ``{MAPEL: achieving global optimality for a non-convex wireless power
  control problem},''
\newblock {\em IEEE Trans. Wirel. Commun.}, vol. 8, no. 3, pp. 1553--1563, Mar.
  2009.

\bibitem{WSRMISO2012}
L.~Liu, R.~Zhang, and K.-C. Chua,
\newblock ``{Achieving global optimality for weighted sum-rate maximization in
  the K-user Gaussian interference channel with multiple antennas},''
\newblock {\em IEEE Trans. Wirel. Commun.}, vol. 11, no. 5, pp. 1933--1945, May
  2012.

\bibitem{MIMOIBCLuo}
Q.~Shi, M.~Razaviyayn, Z.-Q. Luo, and C.~He,
\newblock ``{An iteratively weighted MMSE approach to distributed sum-utility
  maximization for a MIMO interfering broadcast channel},''
\newblock {\em IEEE Trans. Signal Process.}, vol. 59, no. 9, pp. 4331--4340,
  Sep. 2011.

\bibitem{InterferencePricingMIMO}
C.~Shi, D.~A. Schmidt, R.~A. Berry, M.~L. Honig, and W.~Utschick,
\newblock ``{Distributed interference pricing for the MIMO interference
  channel},''
\newblock in {\em Proc. IEEE ICC}, Dresden, Germany, 2009.

\bibitem{Two-cellMIMO}
C.-B. Chae, I.~Hwang, Jr. R.~W.~Heath, and V.~Tarokh,
\newblock ``{Interference Aware-Coordinated Beamforming in a Multi-Cell
  System},''
\newblock {\em IEEE Trans. Wirel. Commun.}, vol. 11, no. 10, pp. 3692--3703,
  Oct. 2012.

\bibitem{BalancingMIMO}
Z.~K.~M. Ho and D.~Gesbert,
\newblock ``{Balancing egoism and altruism on the interference channel: The
  MIMO case},''
\newblock in {\em Proc. IEEE ICC}, Cape Town, South Africa, 2010.

\bibitem{CoorpertiveMIMOIA}
S.~W. Peters and Jr. R.~W.~Heath,
\newblock ``{Cooperative algorithms for MIMO interference channels},''
\newblock {\em IEEE Trans. Veh. Technol.}, vol. 60, no. 1, pp. 206--218, Jan.
  2011.

\bibitem{WSODrawbacks}
I.~Das and J.E. Dennis,
\newblock ``{A closer look at drawbacks of minimizing weighted sums of
  objectives for Pareto set generation in multicriteria optimization
  problems},''
\newblock {\em Struct. and Multidiscip. Optim.}, vol. 14, no. 1, pp. 63--69,
  1997.

\bibitem{MIMOYe}
S.~Ye and R.~S. Blum,
\newblock ``{Optimized signalling for MIMO interference systems with
  feedback},''
\newblock {\em IEEE Trans. Signal Process.}, vol. 51, no. 11, pp. 2839--2848,
  Oct. 2003.

\bibitem{CompetitiveMIMOScutari}
G.~Scutari, D.~P. Palomar, and S.~Barbarossa,
\newblock ``{Competitive design of multiuser MIMO systems based on game theory:
  a unified view},''
\newblock {\em IEEE J. Sel. Areas Commun.}, vol. 26, no. 7, pp. 1089--1103,
  Sep. 2008.

\bibitem{NEvsPareto1}
J.~E. Cohen,
\newblock ``{Cooperation and self-interest: Pareto-inefficiency of Nash
  equilibria in finite random games},''
\newblock {\em Proc. Natl. Acad. Sci. USA}, vol. 95, pp. 9724--9731, Aug. 1998.

\bibitem{VS}
E.~G. Larsson and E.~A. Jorswieck,
\newblock ``{Competition versus cooperation on the MISO interference
  channel},''
\newblock {\em IEEE J. Sel. Areas Commun.}, vol. 26, no. 7, pp. 1059--1069,
  Sep. 2008.

\bibitem{NBMIMOICChen}
Z.~Chen, S.~A. Vorobyov, C.-X. Wang, and J.~Thompson,
\newblock ``{Pareto region characterization for rate control in MIMO
  interference systems and Nash bargaining},''
\newblock {\em IEEE Trans. Automat. Contr.}, vol. 57, no. 12, pp. 3203--3208,
  Dec. 2012.

\bibitem{PricingMonotonicCon}
C.~Shi, R.~A. Berry, and M.~L. Honig,
\newblock ``{Monotonic convergence of distributed interference pricing in
  wireless networks},''
\newblock in {\em Proc. IEEE ISIT}, Seoul, Korea, Jun. 2009.

\bibitem{PricingLocalIMIMO}
C.~Shi, R.~A. Berry, and M.~L. Honig,
\newblock ``{Local interference pricing for distributed beamforming in MIMO
  networks},''
\newblock in {\em Proc. IEEE MILCOM}, Boston, MA, Oct. 2009.

\bibitem{EduardMISOGame}
E.~A. Jorswieck and E.~G. Larsson,
\newblock ``{The MISO interference channel from a game-theoretic perspective: a
  combination of selfishness and altruism achieve Pareto optimality},''
\newblock in {\em Proc. IEEE ICASSP}, Las Vegas, NV, Apr. 2008.

\bibitem{NonlinearAlternatingOp}
D.~P. Bertsekas,
\newblock ``{Nonlinear programming},''
\newblock {\em 2nd ed. Athena Scientific}, 1999.

\bibitem{SeDuMi}
I.~Polik,
\newblock ``{Addendum to the SeDuMi user guide: Version 1.1},''
\newblock 2005,
\newblock available from \texttt{http://sedumi.ie.lehigh.edu.}

\bibitem{CVXTool}
M.~Grant and S.~Boyd,
\newblock ``{CVX: matlab software for disciplined convex programming},''
\newblock 2012,
\newblock available from \texttt{http://cvxr.com/cvx/download/}.

\bibitem{SDPHuang2011}
W.~Ai, Y.~Huang, and S.~Zhang,
\newblock ``{New results on Hermitian matrix rank-one decomposition},''
\newblock {\em Math. Program: Ser. A}, vol. 128, no. 1-2, pp. 253--283, Jun.
  2011.

\bibitem{Charnes-Cooper}
A.~Charnes and W.~W. Cooper,
\newblock ``Programming with linear fractional functions,''
\newblock {\em Naval Res. Logist. Quarter.}, vol. 9, pp. 181--186, 1962.

\bibitem{ConvergenceStationaryP}
M.~V. Solodov,
\newblock ``{On the convergence of constrained parallel variable distribution
  algorithm},''
\newblock {\em SIAM J. Optimization}, vol. 8, no. 1, pp. 187--196, Feb. 1998.

\bibitem{SDRLuo}
Z.-Q. Luo, W.-K. Ma, M.-C. So, Y.~Ye, and S.~Zhang,
\newblock ``{Semidefinite relaxation of quadratic optimization problems},''
\newblock {\em IEEE Signal Processing Mag.}, vol. 27, no. 3, pp. 20--34, May
  2010.

\bibitem{SDPHuang2010}
Y.~Huang, A.~De Maio, and S.~Zhang,
\newblock ``Semidefinite programming, matrix decomposition, and radar code
  design,''
\newblock in {\em Convex Optimization in Signal Processing and Communications},
  Cambridge University Press, New York, 2010, pp. 192--228.

\bibitem{MatrixAnalysisand}
C.~D. Meyer,
\newblock {\em {Matrix analysis and applied linear algebra}},
\newblock Cambridge University Press, New York, NY, USA, 2000.

\bibitem{NonlinearFracProg}
W.~Dinkelbach,
\newblock ``{On nonlinear fractional programming},''
\newblock {\em Management Science}, vol. 13, no. 7, pp. 492--498, 1967.

\bibitem{CA}
K.~Scharnhorst,
\newblock ``Angles in complex vector spaces,''
\newblock {\em Acta Appl. Math.}, vol. 69, no. 1, pp. 95--103, Oct. 2001.

\end{thebibliography}

\begin{IEEEbiography}[{\includegraphics[width=1in,height=1.2in,clip,keepaspectratio]{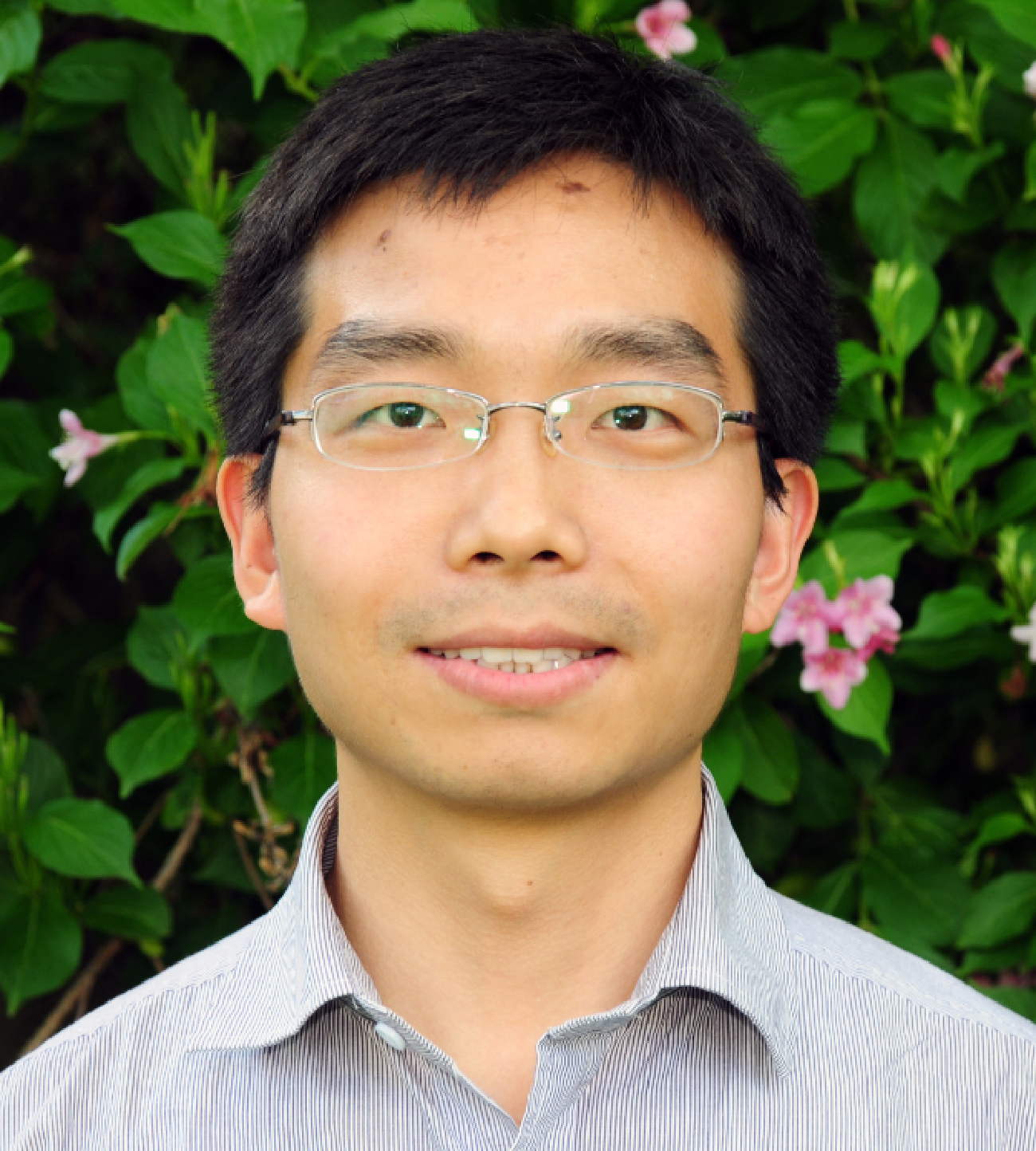}}]{Pan Cao}(S'12) received his B.Eng. degree (with honors) in Department of Mechano-Electronic Engineering in 2008 and M. Eng. degree in Information and Signal Processing in 2011, both from Xidian University, Xi'an, P. R. China. Now, he is with the Communications Laboratory at Dresden University of Technology (TUD), Germany, to peruse his Ph.D. degree under the supervision of Prof. Eduard A. Jorswieck.

His research interests include resource allocation for MIMO interference channels and relay channels with the application of optimization techniques and game theory. He received China Scholarship Council (CSC) Scholarship from 2010 to 2014, the Best Student Paper Award of the 13th IEEE International Workshop on Signal Processing Advances in Wireless Communications (SPAWC) in 2012, and the Qualcomm Innovation Fellowship (QInF) in 2013.
\end{IEEEbiography}
\begin{IEEEbiography}[{\includegraphics[width=1in,height=1.2in,clip,keepaspectratio]{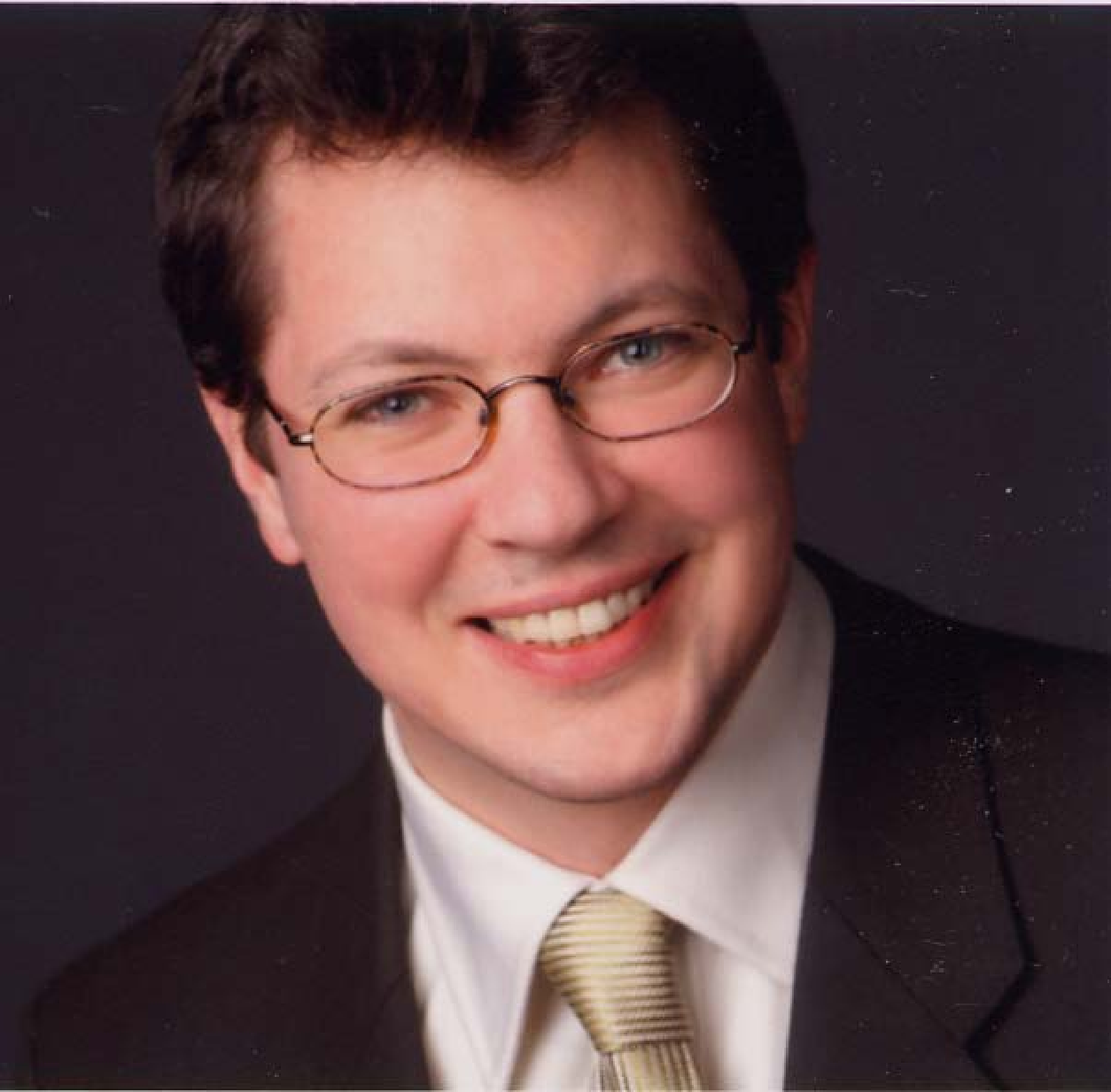}}]{Eduard A. Jorswieck}
(S'01-M'05-SM'08) received his Diplom-Ingenieur degree and Doktor-Ingenieur (Ph.D.) degree, both in electrical engineering and computer science from the Berlin University of Technology (TUB), Germany, in 2000 and 2004, respectively. He was with the Fraunhofer Institute for Telecommunications, Heinrich-Hertz-Institute (HHI) Berlin, from 2001 to 2006. In 2006, he joined the Signal Processing Department at the Royal Institute of Technology (KTH) as a postdoc and became a Assistant Professor in 2007. Since February 2008, he has been the head of the Chair of Communications Theory and Full Professor at Dresden University of Technology (TUD), Germany.

His research interests are within the areas of applied information theory, signal processing and wireless communications. He is senior member of IEEE and elected member of the IEEE SPCOM Technical Committee. From 2008-2011 he served as an Associate Editor and since 2012 as a Senior Associate Editor for IEEE SIGNAL PROCESSING LETTERS. Since 2011 he serves as an
Associate Editor for IEEE TRANSACTIONS ON SIGNAL PROCESSING. Since 2013, Eduard serves as Associate Editor for IEEE TRANSACTIONS ON WIRELESS COMMUNICATIONS. In 2006, he was co-recipient of the IEEE Signal Processing Society Best Paper Award.
\end{IEEEbiography}
\begin{IEEEbiography}[{\includegraphics[width=1in,height=1.2in,clip,keepaspectratio]{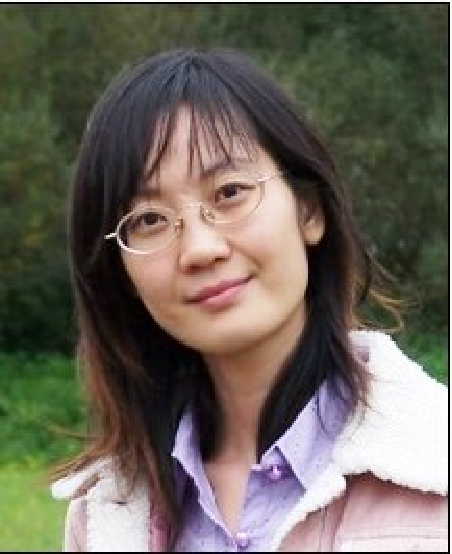}}]{Shuying Shi}
received her Ph.D. degree in electrical
engineering from the Technische Universit\"at Berlin, Germany, in 2009, and her M.Sc. degree in electrical
engineering from the Technische Universit\"at Kaiserslautern, Germany,
in 2002.  She worked at the Fraunhofer German-Sino Lab for Mobile Communications
(MCI), Germany, from 2003 to 2006,   at the Technische Universit\"at
Berlin, Germany, from 2006 to 2009,  at the Link\"oing University, Sweden, from 2009 to 2010, and at the Technische Universit\"at Dresden, Germany, from 2010 to 2011.

Her research interests include  multi-user MIMO signal processing, dynamic resource allocation, and the application of optimization techniques to the design of wireless communications.
\end{IEEEbiography}

\end{document}